\documentclass[12pt, draftclsnofoot, onecolumn]{IEEEtran} 
\pdfoutput=1
\usepackage{changepage}
\usepackage{cite}
\usepackage{amsmath,amssymb,amsfonts}
\usepackage{graphicx}
\usepackage{textcomp}
\usepackage{xcolor}
\usepackage{amsmath}
\usepackage{tabularx}
\usepackage{amsthm}
\usepackage{multirow}
\usepackage{optidef}
\usepackage{hyperref}
\usepackage[ruled,norelsize]{algorithm2e}

\def\ls{\mspace{4mu}} 
\def\wk{\boldsymbol{w}_k}
\def\qn{\boldsymbol{q}_m[n]}
\def\qbs{\boldsymbol{q}_{B}[n]}
\def\hr{H_{R}}
\def\hb{H_{B}}
\def\pm{{p}_m[n]}
\def\pb{{p}_{B}[n]}
\def\rhoo{\rho_{0,m}}
\def\rhobs{\rho_{0,B}}
\def\A{\boldsymbol{X}}
\def\R{\boldsymbol{R}}
\def\Q{\boldsymbol{Q}}
\def\akm{a_{m,k}[n]}
\def\bbm{\beta_{B,m}[n]}
\def\Rk{\bar{R}_k}
\def\Rkn{{R}_k[n]}

\def\Q{\boldsymbol{Q}}
\def\A{\boldsymbol{X}}

\def\P{\boldsymbol{P}}

\def\snrr{\Gamma_{m,k}}
\def\tkm{t_{m,k}}
\def\dkm{d_{m,k}}

\def\wkm{w_{m,k}[n]}
\def\zmn{z_{m}[n]}

\makeatletter
\newcommand{\removelatexerror}{\let\@latex@error\@gobble}
\makeatother

\newtheorem{lemma}{Lemma}
\newtheorem{prop}{Proposition}

\begin{document}
\title{Two-Hop Multi-UAV Relay Network Optimization with Directional Antennas}

\author{Carles Diaz Vilor,~\IEEEmembership{Student Member,~IEEE,}
        and~Hamid~Jafarkhani,~\IEEEmembership{Fellow,~IEEE}
\thanks{The authors are with the Center for Pervasive Communications and
Computing, Department of Electrical Engineering and Computer Science,
University of California, Irvine, CA, 92697 USA (email: $\{$cdiazvil, hamidj$\}$@uci.edu). This work was supported in part by the NSF Award CCF-1815339.}
}

\maketitle

\begin{abstract}
In this paper, we consider the multi-UAV deployment problem for a two-hop relaying system. For a better network performance, UAVs carry directional antennas that are modeled by a realistic radiation pattern. The goal is to maximize the minimum user rates, and therefore achieve fairness in the network. We propose an iterative algorithm to optimize the TDMA scheduling in both hops, UAV trajectories, antenna beamwidths, and transmit power of the base station and relays.  Simulation results show the throughput improvement as a result of optimizing the directional antenna radiation patterns. In addition, we derive the optimal power allocation, which combined with the beamwidth optimization yields to a much better performance.
\end{abstract}


\IEEEpeerreviewmaketitle

\section{Introduction}
Many fields are taking substantial benefits from Unmanned Aerial Vehicles (UAVs) because of their advantages, for example low production cost, easy deployment, control and maneuverability. In particular, using UAVs in wireless communication systems has recently attracted a lot of attention \cite{7995044,Zeng2018TrajectoryMulticasting,7918510,Wu2018JointNetworks,Zhang2018TrajectoryCommunications,8068199,Cheng2018UAVCells,Koyuncu2017DeploymentApproach,7807176,Zeng2017Energy-EfficientOptimization, Zeng2019EnergyUAV}. Some applications of UAV-enabled communications may include the use of UAVs as mobile base stations \cite{7918510, Zeng2018TrajectoryMulticasting,Wu2018JointNetworks},  as well as mobile relays \cite{Zhang2018TrajectoryCommunications,8068199} or for data offloading purposes \cite{Cheng2018UAVCells}. However, to practically include UAVs in the existing ground network, there are still some challenges that need to be addressed. For example, exploiting the UAV mobility requires resource allocation and the design of the trajectory. \par 
While trajectory planning is not an issue in  conventional ground access points,  UAV networks are constrained in terms of flying altitudes, inter-UAV safety distances, non-flying zones, maximum velocity or the on-board energy. For example, UAVs can legally fly only in a determined range of altitudes. Furthermore, to tackle the energy constraint and study energy-efficient deployments, the UAV energy propulsion model is derived for fixed-wing and rotatory-wing UAVs in \cite{Zeng2017Energy-EfficientOptimization} and \cite{Zeng2019EnergyUAV}, respectively. \par 

In addition, wireless networks have experienced a massive densification, mainly to improve the spectral efficiency and to provide new services like 5G and Internet of Things (IoT). One of the premises of 5G is the use of higher frequency bands, such as the millimeter wave band \cite{Zhang2019ANetworks}. Deployments in such frequencies allow more devices per cell at a price of more interference. Although there are many techniques to mitigate interference in ground multi-user scenarios, e.g. look at \cite{Kazemitabar2008MultiuserAntennas,Kazemitabar2009MultiuserAntennas,Koyuncu2012DistributedFeedback} and the references therein, less work has been done in UAV networks in terms of power management and interference coordination, see for example \cite{Wu2018JointNetworks,8811738,8961096}.  While one appealing alternative is the use of directional antennas to avoid the interference to begin with, it is not common to consider such models in UAV networks. In fact, probably the most common assumption is considering isotropic radiation patterns at UAVs \cite{Wu2018JointNetworks,Zeng2018TrajectoryMulticasting,Zhang2018TrajectoryCommunications,Cheng2018UAVCells,8068199,8811738,8961096}. However, in this work, we consider  UAVs that carry directional antennas with the aim of improving the wireless links. \par

Focusing on the works utilizing directional antennas, with the exception of our group's recent papers \cite{Globec,9086619}, a constant-gain directional model has been widely adopted \cite{8379427,7486987,8764580}. Particularly, in our recent conference paper \cite{Globec}, we investigate the optimal 3-D  trajectory optimization problem in a one-hop UAV-enabled down-link scenario while \cite{9086619} investigates the optimal 3-D static  UAV location to obtain a power-efficient deployment.  In contrast to the constant-gain approach, \cite{ConstantineABalanis2016AntennaDesign} suggests a more realistic model, in which the antenna radiation pattern is no longer constant within the dominant direction. In fact, such a model considers a continuous angle-dependent gain given by the cosine-powers of the angle between the source and destination. On the other hand, modeling the radiation pattern as in \cite{ConstantineABalanis2016AntennaDesign} results in complex non-convex optimization problems. However, the networks under consideration present a  considerable advantage, being the dominance of the Line-of-Sight (LoS) channel component under certain conditions \cite{Itu-rGuidelinesServices,LOSLTE,Khuwaja2018ACommunications,8337920,7936620}, that makes the analysis simpler. \par

Finally, an especially interesting use case for UAVs in wireless communications is using them as relays. While static relay networks have been extensively studied in the literature \cite{1056084,1435648,4585342,5437440}, the study of the dynamic relaying is less common. As an example, \cite{Zhang2018TrajectoryCommunications} studies a multi-hop relaying problem with a fixed base station and unitary gain radiation patterns. As another example, \cite{8068199} investigates a two-hop single-UAV deployment with a fixed base station. Both examples  ignore directional antennas. This motivates us to tackle the problem of dynamic relaying using UAVs carrying directional antennas.  Specifically, in this paper, we study a general two-hop multi-UAV dynamic relaying network with directional antennas. More precisely, we formulate the minimum user rate maximization problem for the case where a set of Relay-UAVs (R-UAVs)  are deployed to assist a Mobile Base Station (MBS) to reach the ground users (GUs) in a down-link scenario. 

To deploy such a network, as presented in Figs. \ref{fig:scen1} and  \ref{fig:scen2}, we are interested in the optimization of: (i) resource scheduling in both hops, (ii) UAV trajectories, (iii) antenna beamwidths, and  (iv) transmit power of the MBS and R-UAVs. While each of these optimization problems has been considered for other scenarios in the literature, to the best of our knowledge, there is no attempt in optimizing the combination of them using directional antennas in a two-hop relay network. 

A common formulation of the resource allocation optimization is the classical Time Division Multiple Access (TDMA) scheduling problem for both hops. The UAV trajectory optimization is the subject of many studies in the literature. Sequential Convex Programming (SCP) \cite{Diehl2010RecentEngineering} is one of the most common approaches \cite{Wu2018JointNetworks,Zhang2018TrajectoryCommunications,Cheng2018UAVCells,Zeng2017Energy-EfficientOptimization,Zeng2019EnergyUAV,8811738,8961096,Globec}. 
Other techniques may include graph theory \cite{Zeng2018TrajectoryMulticasting,878915} or artificial intelligence through reinforcement learning methods \cite{8654727,8736350}. 
The beamwidth optimization problem is, in general, non-convex and not well-studied in the literature. We provide the optimal beamwidth  solution and numerical results for some convex cases. As shown through simulation, the insight learned from the convex cases is valid in the non-convex cases as well. Finally,  power optimization for both the MBS and R-UAVs is a key element of interference mitigation in multi-cell networks.
The main energy consumption components of a transceiver node are communication energy and computation energy \cite{HYHJMM}. The experimental measurements show that, in many applications, the computation energy is negligible compared to the communication energy \cite{GMMA}.
Therefore, we study the power allocation problem only considering the communication power and derive analytical expressions for the power allocation of both the MBS and R-UAVs. Hence, the contributions of the paper can be summarized as:
\begin{itemize}
    \item We introduce directional antennas to the general framework of two-hop multi-UAV relaying systems, where UAVs act as both mobile relays and the mobile base station. 
    \item We formulate and solve the maximize minimum GU throughput problem. 
    \item We investigate the UAV's beamwidth optimization problem to improve the throughput. 
    \item We provide  a closed-form analytical solution to the power allocation problem for a given TDMA scheduling, UAV trajectories and beamwidths.
\end{itemize} \par

The remainder of the paper is organized as follows. Section \ref{Sec:SMod} presents the system model of the two-hop multi-UAV relaying scenario. In Section \ref{Sec:Form}, we formulate the maximize minimum GU rate problem. In Section \ref{Sec:Sol}, we divide the original problem into four components and propose an iterative method to solve it. Numerical results are presented and discussed in Section \ref{Sec:Res}. Finally, we provide concluding remarks in Section \ref{Sec:Concl}.

{\em Notation:}  
We write real numbers in $\mathbb{R}$ in small letters. Row vectors are bold. The Euclidean norm of  vector $\boldsymbol{v}$ is given by $||\boldsymbol{v}|| = \sqrt{\sum_n v_n^2}$ and sets are represented in calligraphic letters.

\begin{figure}[!htb]
    \minipage{0.41\textwidth}
    \centering
      \includegraphics[scale=0.40]{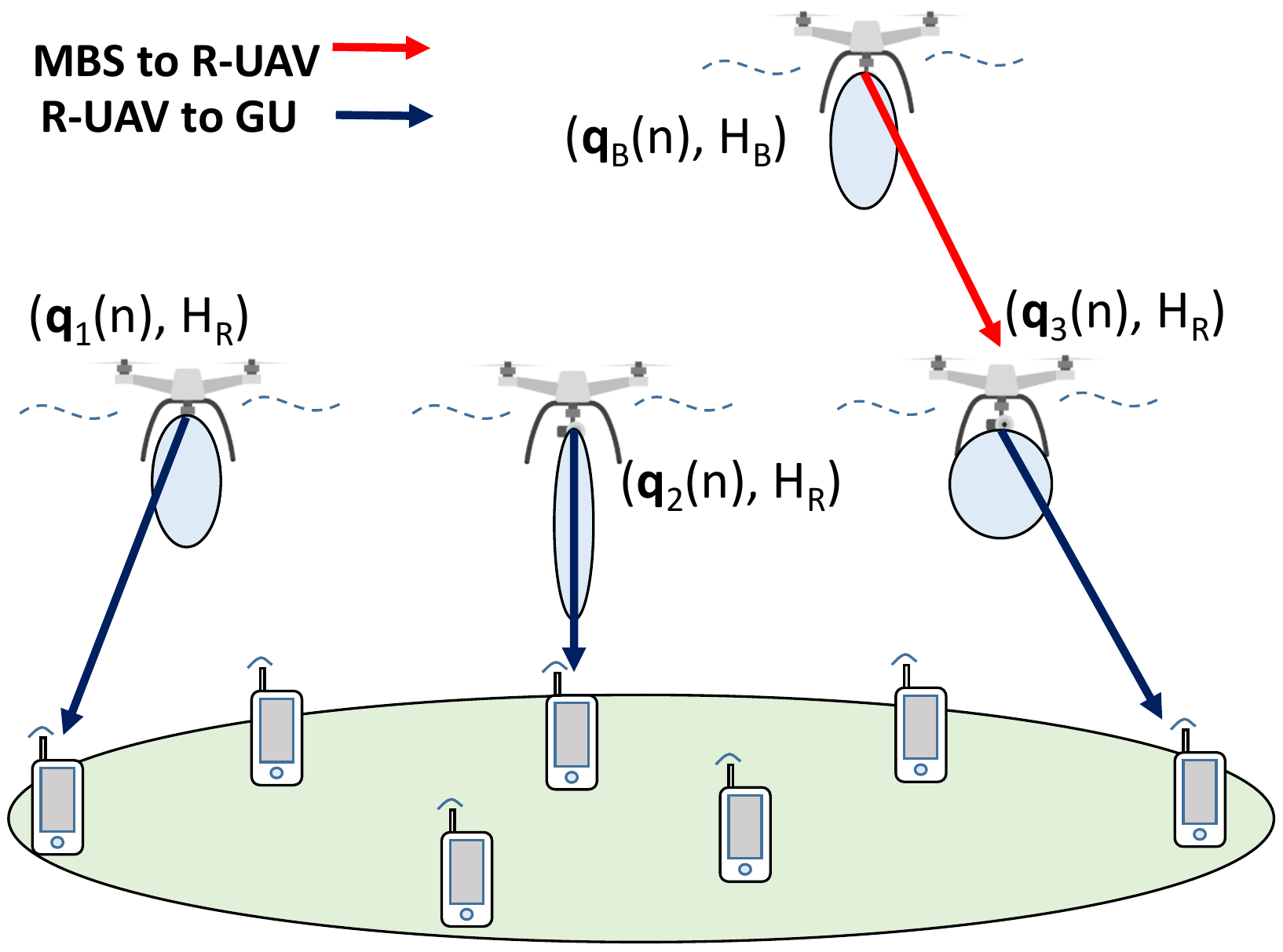}
      \caption{Network structure at snapshot $n$ with one MBS and three R-UAVs.}\label{fig:scen1}
    \endminipage\hfill
    \minipage{0.41\textwidth}%
    \centering
      \includegraphics[scale=0.40]{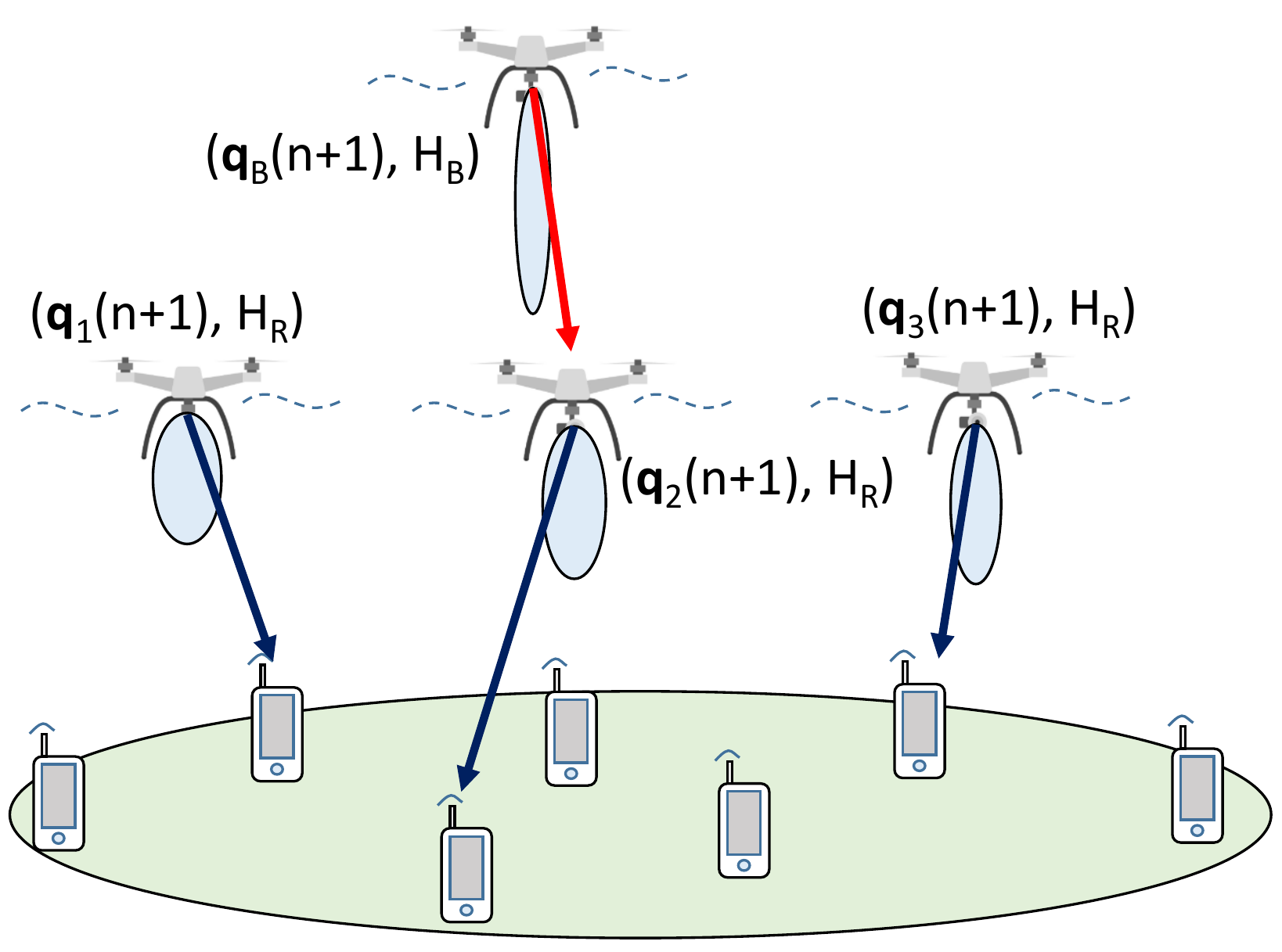}
      \caption{Network structure at snapshot $n+1$ with one MBS and three R-UAVs.}\label{fig:scen2}
    \endminipage
\end{figure}

\section{System Model}\label{Sec:SMod}
As shown in Figs. \ref{fig:scen1} and \ref{fig:scen2}, the networks under consideration feature one source node, named MBS, and $M$ R-UAV relays, enumerated by  $\mathcal{M} = \{1,2,\dots,M \}$,  re-transmitting the information to $K$ GUs, represented by $\mathcal{K} = \{1,2,\dots,K \}$. An example of such a scenario is during a disaster, \cite{7807176}, when the wireless communication infrastructure is down and GUs communicate through the ad-hoc network of UAVs.  In this work, we present and derive relations between the UAV movement, power consumption, channel model and other network features needed to fully deploy such structure with the aim of maximizing the minimum GU rate. \par 
Relay networks have been the subject of research for many years \cite{1056084,1435648,4585342,5437440}, and therefore their advantages are well-known. Probably, the main benefit of relay networks is that when the channel quality between source and destination is not good, we can still find a multi-hop relay path overcoming such channel difficulties  \cite{4373409,4801494}. However, not much work has been done for the case where relays are allowed to move, as in  the case of R-UAVs. In this case, one MBS is used as the source and $M$ R-UAVs are employed to reach the GUs in a two-hop scenario. For ease of exposition, we  refer to the first hop as the links between MBS to R-UAVs and the second hop as the ones from  R-UAVs to GUs. Once the R-UAVs receive data from the MBS, there are four intermediate steps before re-transmitting it: down-conversion, decoding, encoding and up-conversion, which will result in delays, since there is limited hardware and software on board of the UAVs. Therefore, we also take into account a variable signal processing delay. \par

Furthermore, in many existing work in the literature, the location of the BS is assumed to be fixed and is designed to fulfill some requirements, e.g. coverage or data-rate, look at \cite{8302930}  and the references therein. However, such a paradigm can be broken with the use of UAVs, as discussed in the previous section. More particularly and without loss of generality, we assume the BS can be mobile (MBS) for the rest of this work. We denote the time-varying MBS coordinates by a 3-D vector $(\boldsymbol{q_{B}}(t), H_{B})$ where $t$ is the time index between $0$ and the  flying/mission time $T$, sub-index $B$ stands for Base, $\boldsymbol{q_{B}}(t) \in \mathbb{R}^2$ is the ground projection and $H_{B} \in \mathbb{R}$ is the fixed height. Similarly, we can write the set of R-UAV positions as: $(\boldsymbol{q_{m}}(t), H_{R}) \ls \ls m \in \mathcal{M}$. 
Finally, we assume $H_{B} > H_{R}$ and without loss of generality, GUs are located at static positions $\wk{} \in \mathbb{R}^2\ls , \ls \ls k \in \mathcal{K}$. \par

To manage the difficulty of dealing with continuous-time variables, we discretize the time index, $t$, by dividing the time horizon $T$ into $N$ equal time slots, such that 
$T \ls = \ls \delta N$. We also introduce the discrete-time index $n=\frac{t}{\delta}\ls , \ls  \ls n = 1,2,\dots,N$. Therefore, the UAV trajectories of the MBS and the R-UAVs can be expressed by $ (\qbs{}, \hb{})$ and $(\qn{}, \hr$), respectively. 
We refer to $V_{R}$ and $V_{B}$ as the maximum horizontal velocities of the R-UAVs and the MBS, respectively. Then, the first UAV-mobility constraints, referring to the maximum velocity between two generic snapshots, are:
\begin{align}\label{ct:velBS}
    ||\boldsymbol{q}_{B}[n+1] - \qbs{} ||^2 \leq (V_B \delta)^2 \ls \ls \forall n, \mspace{26mu} ||\boldsymbol{q}_{m}[n+1] - \qn{} ||^2 \leq (V_R \delta)^2 \ls \ls \forall n, \ls \forall m.
\end{align}
In addition, we force UAVs to have the same initial and final positions, which in practice means that GUs can be served periodically every $N-1$ time slots:
\begin{align}\label{ct:InitBS}
    \boldsymbol{q}_{B}[1] = \boldsymbol{q}_{B}[N],  \ls \ls \ls \ls \ls \ls \ls \ls \ls \ls \ls \ls \boldsymbol{q}_{m}[1] = \boldsymbol{q}_{m}[N]\ls \ls \forall m.
\end{align}
Finally, to avoid collision between R-UAVs, the following constraint must be satisfied:
\begin{align}\label{ct:Collision}
    || \qn \ls - \ls \boldsymbol{q}_j[n||^2 \geq d^2_{min},  & \ls \ls \ls \forall \ls n,m, j \neq m,
\end{align}
\noindent where $d_{min}$ is defined as the minimum safety distance. The assumption $\hb{} > \hr{}$ makes it impossible to have a  collision between the MBS and R-UAVs. 
\par
The presented constraints have further implications, mainly related to the UAVs' propulsion power consumption. To consider the movement-related power, we assume all UAVs are rotatory-wing UAVs. Apart from the UAV location, the power consumption is affected by other factors such as wind, air density and others;  see \cite{Zeng2019EnergyUAV} for more details. In particular, the total power consumption of a rotatory wing UAV consists of the communication-related power and the movement-related power. As the second term is much higher than the first, we assume that the communication power is nearly constant compared to the propulsion term. Therefore, we approximate the movement power for a rotatory-wing UAV as \cite{Zeng2019EnergyUAV}: 
\begin{align}\label{eq:consumption}
    P_C[n] = P_{0}\bigg( 1 + \frac{3 ||\boldsymbol{v}[n]||^2}{U^2_{tip}} \bigg) + \frac{1}{2}d_0 \rho_a s B  ||\boldsymbol{v}[n]||^3 ,
\end{align}
\noindent where $\boldsymbol{v} = \frac{ \boldsymbol{q}[n+1] - \boldsymbol{q}[n]  }{\delta}$,  $P_{0}$ is the blade-profile power in hovering status constant, $U_{tip}$ represents the speed of the rotor blade, $d_0$ is the fuselage drag, $s$ is the rotor solidity, and $\rho_a$ and $B$ denote the air density and rotor disc area, respectively. We assume the same parameters as in \cite{Zeng2019EnergyUAV}.

For the communication power, we use variables $\pb{}$ and $\pm{}$ for the MBS and R-UAVs  at time $n$, respectively. Both are subject to average and peak power constraints. In particular, for the average terms we have:

\begin{subequations}
    \begin{tabularx}{\textwidth}{XX}
        \begin{equation}\label{ct:powerAvgBS}
        \frac{1}{N} \sum \limits_{n = 1}^N\pb{} \leq P_{B,avg} 
        \end{equation}
        & 
      \begin{equation}\label{ct:powerAvgR}
        \frac{1}{N} \sum \limits_{n = 1}^N \pm{} \leq P_{R,avg} \ls \ls \forall m,
      \end{equation}
    \end{tabularx}
\end{subequations}
\noindent while for the maximum instantaneous transmit power, the constraints are:

\begin{subequations}
    \begin{tabularx}{\textwidth}{XX}
        \begin{align}\label{ct:powerBS}
            \pb{} \leq P_{B,max} \ls \ls \forall n
        \end{align}
        & 
      \begin{align}\label{ct:powerR}
            \pm{} \leq P_{R,max} \ls \ls \forall m,n.
        \end{align}
    \end{tabularx}
\end{subequations}
Furthermore, Air-to-Air (A2A) and Air-to-Ground (A2G) channel modeling  is an active research topic. In fact, channel measurements in many practical scenarios, such as rural, have shown that both A2A and A2G communications follow a free space path-loss model  \cite{Itu-rGuidelinesServices,LOSLTE,Khuwaja2018ACommunications}. Such assumption is subject to the fact that UAVs should fly at a considerable altitude, since the probability of being in LoS mainly depends on the distance and altitude of the UAV transceiver \cite{8337920},\cite{7936620}.  Therefore, with the aim of providing essential insight and under the premise that the involved UAVs meet such conditions, we assume the wireless channels are mainly dominated by LoS links. A possible extension to Non-LoS (N-LoS) channels is left as future work. We also assume possible Doppler mismatches caused by the UAV dynamics  are compensated at the receiver, as well as the existing asynchrony  between the involved clocks. Therefore,  the channel gain from the MBS to the $m$-th R-UAV at time $n$ is given by: 
\begin{align}\label{eq:channel1}
    h_{B,m}[n] = \frac{A d_{0,A}^\kappa}{ (||\qbs{} - \qn||^2 \ls + \ls ( \hb{} - \hr{})^{2} )^{ \frac{\kappa}{2}} },
\end{align}
\noindent where $A$ is a unit-less constant depending on the antenna characteristics, $d_{0,A}$ is a reference distance for the A2A channel and $\kappa \geq 1$ refers to the  path-loss exponent. Similarly, we can define the channel gain between the $m$-th R-UAV and the $k$-th GU as:
\begin{align}\label{eq:channel2}
    h_{m,k}[n] = \frac{A d_{0,G}^\kappa}{ (|| \qn - \wk||^2 \ls + \ls \hr{}^{2} )^{ \frac{\kappa}{2}} },
\end{align}
\noindent where $d_{0,G}$ is the reference distance for the A2G channel. Furthermore, we adopt the notion of angle-dependent antenna gains for UAV  optimization problems in which UAVs are equipped with directional antennas \cite{Globec,9086619}.  At time $n$, the antenna gain of the MBS transmitting to the $m$-th R-UAV, denoted by $G_B(\theta_{B,m}[n])$, and the antenna gain of the $m$-th R-UAV transmitting to the $k$-th GU, denoted by $G_R(\theta_{m,k}[n])$, are:
\begin{align}\label{eq:antenna1}
\begin{small}
    G_B(\theta_{B,m}[n])  = D_o(r_B[n])cos^{r_B[n]}(\theta_{B,m}[n])  = D_0(r_B[n]) \frac{| \hb{} - \hr{} |^{r_{B}[n]}}{(|| \qbs{} - \qn||^2  +  ( \hb{} - \hr{} )^{2} )^{ \frac{r_B[n]}{2}} }, 
\end{small}
\end{align}
\begin{align}\label{eq:antenna2}
    G_R(\theta_{m,k}[n]) \ls = \ls D_o(r_m[n])cos^{r_m[n]}(\theta_{m,k}[n]) = D_0(r_m[n]) \frac{\hr{}^{r_{m}[n]}}{(|| \qn - \wk||^2 \ls + \ls  \hr{}^{2} )^{ \frac{r_m[n]}{2}} }. 
\end{align}
The model depends on parameters $ \{r_B[n] , r_m[n] \} \geq 1$, which define the maximal directivity of the antenna at  $\theta = 0$ as $D_0(r) = 2 (r + 1)$  \cite{ConstantineABalanis2016AntennaDesign}. Note that $r = 0$ is the same as having an isotropic antenna and, for simplicity, we ignore side-lobes,  represented by $cos(l\theta)$ patterns. The larger the parameter $r$, the narrower the beam and therefore the directivity of the antenna would increase. Therefore, if we are interested in covering a precise area, it is better to use narrow beams. As previously mentioned, we allow the UAVs to optimize their beamwidths $\{ r_B[n], r_m[n] \}$. 
Since UAV transceivers have limited hardware capacity to switch the beamwidth, we limit the range of possibles values through the following inequalities:
\begin{align}\label{eq:degreeval}
    r_{min} \leq \{ r_B[n], r_m[n] \} \leq r_{max } \ls \ls \forall m, n.
\end{align}
On the other hand, and for the sake of simplicity, the receiver antennas placed at R-UAVs and GUs are assumed to have a unitary power gain. However, we could add such directional patterns at the receiver side as well with minor modifications in our formulation. \par

In order not to overload the formulation, we define the  terms $\rhobs{}  = 2\ls A \ls d_{0,A}^\kappa $ and $\rhoo{} =  2\ls A \ls d_{0,G}^\kappa$. Thus, we can express the instantaneous rate from the MBS to the $m$-th R-UAV as:
\begin{align}\label{eq:rateBR}
    R_{B,m}[n] =  \log_2 \bigg( 1 + \frac{\pb{} \rhobs{} (r_B[n] + 1)| \hb{} - \hr{} \ls |^{r_{B}[n]} }{ \sigma^2 (||\qbs{} - \qn||^2 + (  \hb{} - \hr{})^{2} )^{ \frac{\kappa + r_B[n]}{2}  } }  \bigg),
\end{align}
\noindent where $\sigma^2$ is the noise power at the receiver side, following a  circularly symmetric complex Gaussian distribution $\mathcal{CN}(0, \sigma^2)$. Similarly, at a generic snapshot, the instantaneous rate from the $m$-th R-UAV to the $k$-th GE is:
\begin{align}\label{eq:rateRU}
    R_{m,k}[n] =   \log_2 \bigg( 1 + \frac{\pm{} \rhoo{} (r_m[n]+1)\hr{}^{r_{m}[n]}}{ \sigma^2 (||\qn - \wk||^2 + \hr{}^{2} )^{ \frac{\kappa + r_m[n]}{2}  } }  \bigg).
\end{align}
Furthermore, we use a TDMA in both hops. For simplicity, we assume this system is capable of mitigating interference, either because there is enough bandwidth ($M + 1$ orthogonal channels are required) or because smart re-using techniques are utilized to ensure that the distances between links re-using the same channel are large enough to make the co-channel interference negligible. We introduce variables 
$\beta_{B,m}[n]$ and $a_{k,m}[n]$ representing the TDMA scheduling in the first and second hops, respectively. As a consequence, since the MBS can only serve one R-UAV at each time, apart from the binary assumption on $\bbm{}$, the following constraint must be met:
\begin{align}\label{ct:Beta1}
     0 \leq\sum_{m} \bbm{} \leq 1  \ls \ls \forall n.
\end{align}
Besides, a fixed R-UAV can serve only one GU and a fixed GU can only be served by one R-UAV. As a result, two more constraints appear in our formulation, ensuring a one-to-one mapping between R-UAVs and GUs:

\begin{subequations}
    \begin{tabularx}{\textwidth}{XX}
        \begin{align}\label{ct:Alpha1}
             0 \leq\sum_{k}\akm \leq 1  \ls \ls \forall m,n,
        \end{align}
        & 
      \begin{align}\label{ct:Alpha2}
            0 \leq \sum_{m}\akm \leq 1  \ls \ls \forall k,n.
        \end{align}
    \end{tabularx}
\end{subequations}
Finally, routing and relaying problems are subject to causality constraints, meaning that a given router/relay cannot forward any information that has not previously arrived.
Hence, each R-UAV needs a sufficiently large buffer to store the information from the MBS until it is capable of re-transmitting it to the GUs. In this work, we assume the buffer has been previously designed and has enough memory for relaying purposes. Many authors assume a processing time of one time slot \cite{Zhang2018TrajectoryCommunications}. However, to analyze the consequences of such delay, we consider a general and deterministic processing time  $D \geq 0$. Therefore, the following causality constraint must be taken into account:
\begin{align}\label{ct:Causality}
    \sum \limits_{i = 1}^{n-D} \beta_{B,m}[i] R_{B,m}[i] \geq \sum \limits_{i = D+1}^{n} \sum \limits_{k = 1}^{K}  a_{m,k}[i]R_{m,k}[i] \mspace{20mu} \forall m, \ls \ls n = D + 1,\dots,N.
\end{align}
As a result, the instantaneous rate of the $k$-th GU  is:
\begin{align}
    R_k[n] = \sum \limits_{m = 1}^M \akm R_{m,k}[n],
\end{align}
\noindent and its averaged value is given by $\Rk = \frac{1}{N} \sum \limits_{n = D+1}^N \Rkn{}$.

\section{Problem Formulation}\label{Sec:Form}
The goal of this work is to maximize the minimum users' rate in a UAV relay network in which both the MBS and the R-UAVs are allowed to change their position over time. However, the UAV movements are subject to physical constraints, as presented in \eqref{ct:velBS}-\eqref{ct:Collision}. Furthermore, UAVs must satisfy a lifetime constraint that considers the total power of the battery. From the communications perspective, the problem is subject to communication-related power constraints, \eqref{ct:powerAvgBS}-\eqref{ct:powerR}. Since we allow to adapt the beamwidth of the UAV antennas,  we include constraint \eqref{eq:degreeval} for the minimum and maximum directivity degrees. In addition, we mentioned the TDMA rules in both hops  \eqref{ct:Beta1}-\eqref{ct:Alpha2} and the need of the causality constraint \eqref{ct:Causality}.  To this end, the optimization variables include the TDMA scheduling association in both hops, represented by $\A = \{ \bbm{} \ls \ls, \ls \ls  \akm \ls \ls \forall m,k,n \}$, the 2-D position of the MBS and R-UAVs over all time slots, denoted by $\Q{} = \{ \qbs{} \ls , \ls \qn \ls \ls \forall m, n \}$, the beamwidths of each UAV  $\boldsymbol{R} = \{ r_B[n]  \ls , \ls  r_m[n] \ls \ls \forall m,n  \}$ and, finally, the transmit power of both the MBS and R-UAVs, given by the set $\P{} = \{ \pb{} \ls , \ls \pm{} \ls \ls \forall m,n \}$. As a result, denoting $\mu \ls = \ls \min_k \ls \Rk$, the optimization problem can be formulated as:
\begin{equation}
    \begin{aligned}
    & \underset{\mu, \A, \Q, \R, \P}{\text{max}}
    & & \mu \\
    & \text{s.t.} & &  \Rk \geq \mu \ls \ls \ls \forall k \\
    & & &  \akm, \bbm{} \in \{0,1\} \ls \ls \ls \forall m,k,n \\
    & & & \sum \limits_{n = 1}^N P_{C,B}[n]  \leq P_{UAV} \mspace{10mu} , \mspace{10mu} \sum \limits_{n = 1}^N P_{C,m}[n]  \leq P_{UAV}, \ls \ls \forall m  \\
    & & &\eqref{ct:velBS}- \eqref{ct:Collision}, \eqref{ct:powerAvgBS}- \eqref{ct:powerR}, \eqref{eq:degreeval}, \eqref{ct:Beta1}- \eqref{ct:Causality}
    \end{aligned}
\end{equation}

\noindent where $P_{C,B}[n]$ and $P_{C,m}[n]$ refer to the movement-related power consumption of the MBS and R-UAVs, respectively, and $P_{UAV}$ is the total trajectory-related stored power in the batteries. In addition, such a problem presents two main issues, making it challenging and difficult to solve. First, the binary nature of $\bbm{}$ and $\akm$ and the corresponding integer constraints result in an NP-hard problem. To make it more tractable, we  relax the integer constraint as follows:
\begin{align}\label{ct:alphareal}
    0 \leq \{ \bbm{} \ls , \ls \akm \} \leq 1 \ls \ls \ls \forall m,k,n.
\end{align}
Second, the non-convexity of many constraints with respect to the trajectories and beamwidths adds to the complexity of the problem. Therefore, it is desirable to develop a more tractable formulation. Hence, in light of such challenges, we propose a method in which we split the original problem into four sub-problems and solve them separately in the next section.

\section{Proposed Solution}\label{Sec:Sol}
As mentioned, since we have four sets of optimization variables, we solve four different sub-problems: (i) TDMA scheduling optimization with fixed UAV trajectories, beamwidths, and transmit power;  (ii) UAV trajectory optimization with fixed TDMA scheduling, beamwidths, and transmit power; (iii) beamwidth optimization with fixed TDMA scheduling, UAV trajectories, and transmit power; and (iv) power optimization with fixed TDMA scheduling, UAV trajectories, and beamwidths. Once the solution of each problem is obtained separately, we apply a Block Coordinate Descent (BCD) method to iteratively maximize the minimum user rate until convergence  \cite{BCD}.

\subsection{TDMA Scheduling Optimization Sub-Problem}
First, we solve the TDMA association in both hops, given by the set of variables $\A$. For fixed  $\Q$, $\R$ and $\P$, such problem, named (O-X1), can be formulated as:
\begin{equation}
    \begin{aligned}
    & \underset{\mu, \A}{\text{max}}
    & & \mu \\
    & \text{s.t.} & &  \Rk \geq \mu \\
    & & & \eqref{ct:Beta1} - \eqref{ct:Causality} , \eqref{ct:alphareal}
    \end{aligned}
\end{equation}
Since the objective function and the constraints are linear with respect to the optimization variables, we can efficiently solve it using standard linear programming (LP) techniques, such as the interior point method \cite{Boyd2004ConvexOptimization}. A method to reconstruct the solution of (O-X1) into a binary scheduling without compromising optimality has been studied in \cite{Wu2018JointNetworks}.

\subsection{UAV Trajectory Optimization Sub-Problem}\label{subsect:power}
Now, we solve the sub-problem that relates the trajectories of the MBS and R-UAVs. We consider fixed values for $\A$, $\R$ and $\P$. Therefore, the UAV trajectories can be optimized by means of solving the next problem, named (O-T1):
\begin{equation}
    \begin{aligned}
    & \underset{\mu, \Q}{\text{max}}
    & & \mu \\
    & \text{s.t.} & &  \Rk \geq \mu \ls \ls \forall k \\
    & & & \sum \limits_{n = 1}^N P_{C,B}[n]  \leq P_{UAV} \mspace{10mu} , \mspace{10mu} \sum \limits_{n = 1}^N P_{C,m}[n]  \leq P_{UAV}, \ls \ls \forall m  \\
    & & &\eqref{ct:velBS}- \eqref{ct:Collision},  \eqref{ct:Causality}
    \end{aligned}
\end{equation}
Note that the terms $\Rk{}$ are non-convex with respect to $\qbs{}$. Furthermore, \eqref{ct:Collision} and \eqref{ct:Causality} are non-convex constraints as well. Consequently, (O-T1) is a non-convex optimization problem, hard to solve and without a general technique to obtain the global optima. To handle it, we first reformulate (O-T1) as an equivalent sub-problem, (O-T2). Afterwards, we apply the SCP technique \cite{Diehl2010RecentEngineering}  to solve it. First, to simplify the presentation, let us define the following non-trajectory dependent terms: 
\begin{align*}
    \Gamma_{B,m}[n] = \frac{ \pb{} \rhobs{} (r_B[n] + 1)| \hb{} - \hr{} \ls |^{r_{B}[n]}}{ \sigma^2  }, \mspace{24mu} \Gamma_{m,k}[n] =\frac{ \pm{} \rhoo{}(r_m[n]+1)\hr{}^{r_{m}[n]}}{ \sigma^2 }.
\end{align*}
Then, we define the set of slack variables $\boldsymbol{D} = \{ \dkm{}[n], \ls \ls n = D + 1,\dots, N \ls \ls \forall k,m  \} $. Consequently, we can reformulate (O-T1) as an equivalent problem, named (O-T2):

\begin{equation*}
    \begin{aligned}
    & \underset{\mu, \Q, \boldsymbol{D}}{\text{max}}
    & & \mu \\
    & \text{s.t.} & &  \sum \limits_{n = D+1}^N \sum \limits_{m = 1}^M \akm \dkm{}[n] \geq \mu \ls \ls \forall k  \\
    & & &\sum \limits_{i = 1}^{n-D} \beta_{B,m}[i] R_{B,m}[i] \geq \sum \limits_{i = D+1}^{n} \sum \limits_{k = 1}^{K} a_{m,k}[i] \dkm{}[i] \mspace{30mu} n = D+1,\dots,N \ls , \ls \forall m\\
    & & & {{\dkm{}[n] \leq \log_2 \bigg( 1 + \frac{\Gamma_{m,k}[n]}{   (||\qn - \wk||^2 + \hr{}^{2} )^{ \frac{r_m[n]+\kappa}{2}}    }  \bigg)}} \ls \ls \mspace{20mu} n = D+1,\dots,N \ls , \ls \forall k, m\\
    & & & \sum \limits_{n = 1}^N P_{C,B}[n]  \leq P_{UAV} \mspace{10mu} , \mspace{10mu} \sum \limits_{n = 1}^N P_{C,m}[n]  \leq P_{UAV}, \ls \ls \forall m  \\
    & & &\eqref{ct:velBS}- \eqref{ct:Collision} 
    \end{aligned}
\end{equation*}

\begin{lemma}\label{Lemma1}
 (O-P2) is equivalent to (O-P1).
\end{lemma}
\begin{proof}
The proof can be found in Appendix \ref{App:ProofLema1}.
\end{proof}

However, since the instantaneous rates and the collision-avoidance terms are non-convex with respect to $\qbs{}$ and $\qn{}$, it is still a challenging problem to solve. \par
In the following, we develop a method based on the SCP technique. SCP algorithms alternate between two steps: (i) approximate the non-convex terms by convex terms, providing an approximated problem and (ii) optimally solve the approximated problem until convergence. However,  both instantaneous rates, the one from the MBS to R-UAVs and the one from R-UAVs to GUs, are convex with respect to $||\qbs{} - \qn||^2$ and $|| \qn - \wk||^2$, respectively. Since every convex function is lower bounded by its first order Taylor expansion, we can obtain a lower bound on both rates as presented in the following terms:
\begin{multline}\label{eq:Bound1}
    R_{B,m}[n] 
    \geq -A_{B,m}^p[n]\big(  ||\qbs{} - \qn{}||^2 - ||\boldsymbol{q}_{B}^p[n] - \boldsymbol{q}_m^p[n]||^2 \big) +  B_{B,m}^p[n] = R_{B,m}^{lb}[n],
\end{multline}
where 
\begin{align*}
    A_{B,m}^p[n] = \frac{\frac{(\kappa + r_B[n])}{2} \log_2(e)}{ ( ||\boldsymbol{q}_{B}^p[n]-\boldsymbol{q}_m^p[n] ||^2 \ls + \ls (\hb{} - \hr{})^2 ) } 
    \cdot \frac{1}{\bigg(  1 + \frac{\big(||\boldsymbol{q}_{B}^p[n]-\boldsymbol{q}_m^p[n] ||^2 \ls + \ls (\hb{} - \hr{})^2 \big)^{\frac{\kappa + r_B[n]}{2}} }{\Gamma_{B,m}[n]}\bigg)}, 
\end{align*}
and
\begin{align*}
    B_{B,m}^p[n] = R_{B,m}[n] \bigg|_{\qbs{} = \boldsymbol{q}_{B}^p[n] \ls, \ls \qn = \boldsymbol{q}_m^p[n] },
\end{align*}
\noindent where super-index $p$ refers to the point where the function is approximated by its lower bound. Proceeding in the same manner for $R_{m,k}[n]$, we obtain:
\begin{multline}\label{eq:Bound2}
    R_{m,k}[n] 
    \geq- C_{m,k}^p[n]\big(  ||\qn{} - \wk||^2 - ||\boldsymbol{q}_{m}^p[n] - \wk||^2   \big) +  D_{m,k}^p[n] = R_{m,k}^{lb}[n],
\end{multline}
where
\begin{align*}
    \small{C_{m,k}^p[n] = \frac{\frac{(\kappa + r_m[n])}{2} \log_2(e)}{ ( ||\boldsymbol{q}_m^p[n]  - \wk{} ||^2 \ls + \ls \hr{}^2 ) } 
    \cdot \frac{1}{\bigg(  1 + \frac{\big(||\boldsymbol{q}_m^p[n] - \wk ||^2 \ls + \ls \hr{}^2 \big)^{\frac{\kappa + r_m[n]}{2}}}{\Gamma_{m,k}[n]} \bigg)} \ls \ls , \ls \ls D_{m,k}^p[n] = R_{m,k}[n] \bigg|_{\qn = \boldsymbol{q}_m^p[n] }.}
\end{align*}

Finally, for the collision-avoidance constraint, after applying SCP, we obtain:
\begin{align}\label{ct:collision2}
    - 2(\boldsymbol{q}_m^p[n]-\boldsymbol{q}_j^p[n])^T(\boldsymbol{q}_m[n]-\boldsymbol{q}_j[n]) 
    -|| \boldsymbol{q}_m^p[n]  -  \boldsymbol{q}_j^p[n]||^2  \geq d^2_{min,h}  \ls \ls \ls \forall \ls n,m, j \neq m.
\end{align}
With any given local points, $\boldsymbol{q}_m^p[n]$ and $\boldsymbol{q}_B^p[n]$, and the lower bounds obtained in \eqref{eq:Bound1} and \eqref{eq:Bound2}, we can formulate the next convex optimization problem, named (O-T3):
\begin{equation*}
    \begin{aligned}
    & \underset{\mu, \Q, \boldsymbol{D}}{\text{max}}
    & & \mu \\
    & \text{s.t.} & &  \sum \limits_{n = D+1}^N \sum \limits_{m = 1}^M \akm \dkm{}[n] \geq \mu \ls \ls \forall k  \\
    & & &\sum \limits_{i = 1}^{n-D} \beta_{B,m}[i] R_{B,m}^{lb}[i] \geq \sum \limits_{i = D+1}^{n} \sum \limits_{k = 1}^{K} a_{m,k}[i] \dkm{}[i] \mspace{30mu} n = D+1,\dots,N \ls , \ls \forall m\\
    & & & \dkm{}[n] \leq R_{m,k}^{lb}[n] \mspace{30mu} n = D+1,\dots,N \ls \forall k,m \\
    & & & \sum \limits_{n = 1}^N P_{C,B}[n]  \leq P_{UAV} \mspace{10mu} , \mspace{10mu} \sum \limits_{n = 1}^N P_{C,m}[n]  \leq P_{UAV}, \ls \ls \forall m  \\
    & & &\eqref{ct:velBS}, \eqref{ct:collision2}  
    \end{aligned}
\end{equation*}

Since all constraints in (O-T3) are jointly convex with respect to $\qn{}$, $\qbs{}$ and $\dkm{}[n]$, we  conclude that (O-T3) is a convex optimization problem, and therefore can be solved by standard optimization solvers, such as CVX \cite{Boyd2004ConvexOptimization}. As a result, the optimal value obtained for (O-T3) serves as a lower bound on the optimal solution of the original problem, (O-T1).

\subsection{Beamwidth Optimization Sub-Problem}\label{subsect:beam}
In this section, we focus on the third sub-problem, which tries to attain the optimal values for the directivity degrees, $r_B[n]$ and $r_m[n]$, of both the MBS and R-UAVs given a fixed TDMA scheduling, UAV trajectories, and power allocation. Such a problem is, in general, non-convex. However, we provide a 
discussion for the convex case, as the same idea applies to the results obtained from the general non-convex formulation in Section \ref{Sec:Res}. For such a case, the optimal beamwidth degrees for the MBS and R-UAVs need to be extracted from a highly non-linear equation of the type: 
\begin{eqnarray}\label{eq:levcurv}
   \frac{ \cos^{r}{(\theta}) \bigg( (r+1) \log( \cos(\theta)) + 1 \bigg)  }{   (r + 1 ) \cos^{r}{(\theta)} + \frac{1}{\Gamma^{'}}      } =  K,
\end{eqnarray}
\noindent where for simplicity, we have dropped all sub-scripts and time indices,  $K$ depends on the Lagrangian multipliers and TDMA variables, and $\Gamma^{'}$ is the equivalent channel from transmitter to receiver. Recall that $\theta$ represents the elevation angle between source and destination. 
For a given $K$ and $\Gamma^{'}$, an analytical expression for the optimal directivity value $r$ in terms of $\theta$ is not available. However, numerical solutions for $\Gamma^{'} = 10$ and $\Gamma^{'} = 50$ given different values of $K$ are provided in Figs.  \ref{fig:curv1} and  \ref{fig:curv2}, respectively. As can be seen in these figures, given a fixed value of $K$, in general, for low elevation angles, the optimal $r$ tends to be high. A low elevation angle corresponds to a source that  flies nearly on top of the destination. In such a case, the source increases the value of $r$ in order to increase the directivity with a narrower and more focused beam. On the other hand, when the elevation angle is high, i.e.,  the source does not fly near the destination, the tendency is to decrease the value of $r$ and therefore have a wider beam to reach the destination with a less directive pattern. The ranges of elevation for which such assumption is valid are mainly determined by the curves in \eqref{eq:levcurv} that present only one solution, e.g. $0 \leq \theta \leq 55$ for $\Gamma^{'} = 10$ and $0 \leq \theta \leq 60$ for $\Gamma^{'} = 50$, values easy to attain due to the UAV flying altitudes. \par

Going back to the general non-convex beamwidth optimization problem, we can write such a problem, named (O-R1), as follows:
\begin{equation}\label{opt:Degree}
    \begin{aligned}
    & \underset{\mu, \R}{\text{max}}
    & & \mu \\
    & \text{s.t.} & &  \Rk \geq \mu \ls \ls \forall k\\
    & & & \eqref{eq:degreeval},  \eqref{ct:Causality}
    \end{aligned}
\end{equation}
To make equations more manageable, we define the following constant terms:
\begin{align*}
    \Gamma_{B,m}^{'}[n] =\frac{ \pb \rhobs{}{}}{ \sigma^2 (||\qbs{} - \qn{}||^2 + (\hb{}-\hr{})^{2} )^{ \frac{\kappa}{2}  }}, \mspace{16mu}  \Gamma_{m,k}^{'}[n] =\frac{ \pm{} \rhoo{}}{ \sigma^2 (||\qn - \wk||^2 + \hr{}^{2} )^{ \frac{\kappa}{2}  }}.
\end{align*}

As mentioned, (O-R1)  is in general non-convex. We apply the  Sequential Linear Programming (SLP) technique, consisting of two steps: (i) linearizing the non-convex functions and (ii) solving the LP problem. Consequently, we can formulate the following LP problem, named (O-R2):
\begin{equation*}
    \begin{aligned}
    & \underset{\mu, \R }{\text{max}}
    & & \mu \\
    & \text{s.t.} & &  \sum \limits_{n = D+1}^N \sum \limits_{m = 1}^M \akm R_{m,k}^L[n] \geq \mu \ls \ls \forall k  \\
    & & &\sum \limits_{i = 1}^{n-D} \beta_{B,m}[i] R_{B,m}^L[i] \geq \sum \limits_{i = D+1}^{n} \sum \limits_{k = 1}^{K} a_{m,k}[i] R_{m,k}^L[i] \mspace{20mu} n = D+1,\dots,N \ls , \ls \forall m\\
    & & & \max( r_{min}, r_B^p[n] - \epsilon) \leq r_B[n] \leq \min( r_{max}, r_B^p[n] + \epsilon) \mspace{20mu} \forall n\\
    & & & \max( r_{min}, r_m^p[n] - \epsilon) \leq r_m[n] \leq \min( r_{max}, r_m^p[n] + \epsilon) \mspace{20mu} \forall n,m
    \end{aligned}
\end{equation*}
where the last two constraints make sure the linear approximation is tight enough around $r_B^p[n]$ and $r_m^p[n]$ by means of the parameter $\epsilon$. In addition, $R_{B,m}^L[n]$ and $R_{m,k}^L[n]$ are defined as the first order Taylor expansion with respect to $r_B[n]$ and $r_m[n]$, respectively:
\begin{eqnarray*}
    R_{B,m}^L[n] = E_{B,m}^p[n] (r_B[n] - r_B^p[n]) + F_{B,m}^p[n] , \mspace{8mu} R_{m,k}^L[n] = G_{m,k}^p[n] (r_m[n] - r_m^p[n]) + H_{m,k}^p[n],
\end{eqnarray*}
where:

\begin{small}
\begin{align*}
    \begin{small}
    E_{B,m}^p[n] = \frac{ \Gamma_{B,m}^{'}[n] \cos( \theta_{B,m}[n])^{r_B^p[n]} \big( (r_B^p[n] + 1)\log( \cos( \theta_{B,m}[n]))+1\big)}{ \log(2) \big(  \Gamma_{B,m}^{'}[n] (r_B^p[n] + 1)\cos( \theta_{B,m}[n])^{r_B^p[n]} + 1 \big)} , \mspace{8mu} F_{B,m}^L[n] = R_{B,m}[n] \bigg|_{r_B[n] = r_B^p[n] },
    \end{small}
\end{align*}
\end{small}
and
\begin{small}
\begin{eqnarray*}
    G_{m,k}^p[n] = \frac{ \Gamma_{m,k}^{'}[n] \cos( \theta_{m,k}[n])^{r_m^p[n]} \big( (r_m^p[n] + 1)\log( \cos( \theta_{m,k}[n]))+1\big)}{ \log(2) \big(  \Gamma_{m,k}^{'}[n] (r_m^p[n] + 1)\cos( \theta_{m,k}[n])^{r_m^p[n]} + 1 \big)} , \mspace{8mu} H_{m,k}^L[n] = R_{m,k}[n] \bigg|_{r_m[n] = r_m^p[n] }.
\end{eqnarray*}
\end{small}

\begin{figure}[!htb]
    \minipage{0.48\textwidth}
      \centering
      \includegraphics[scale=0.57]{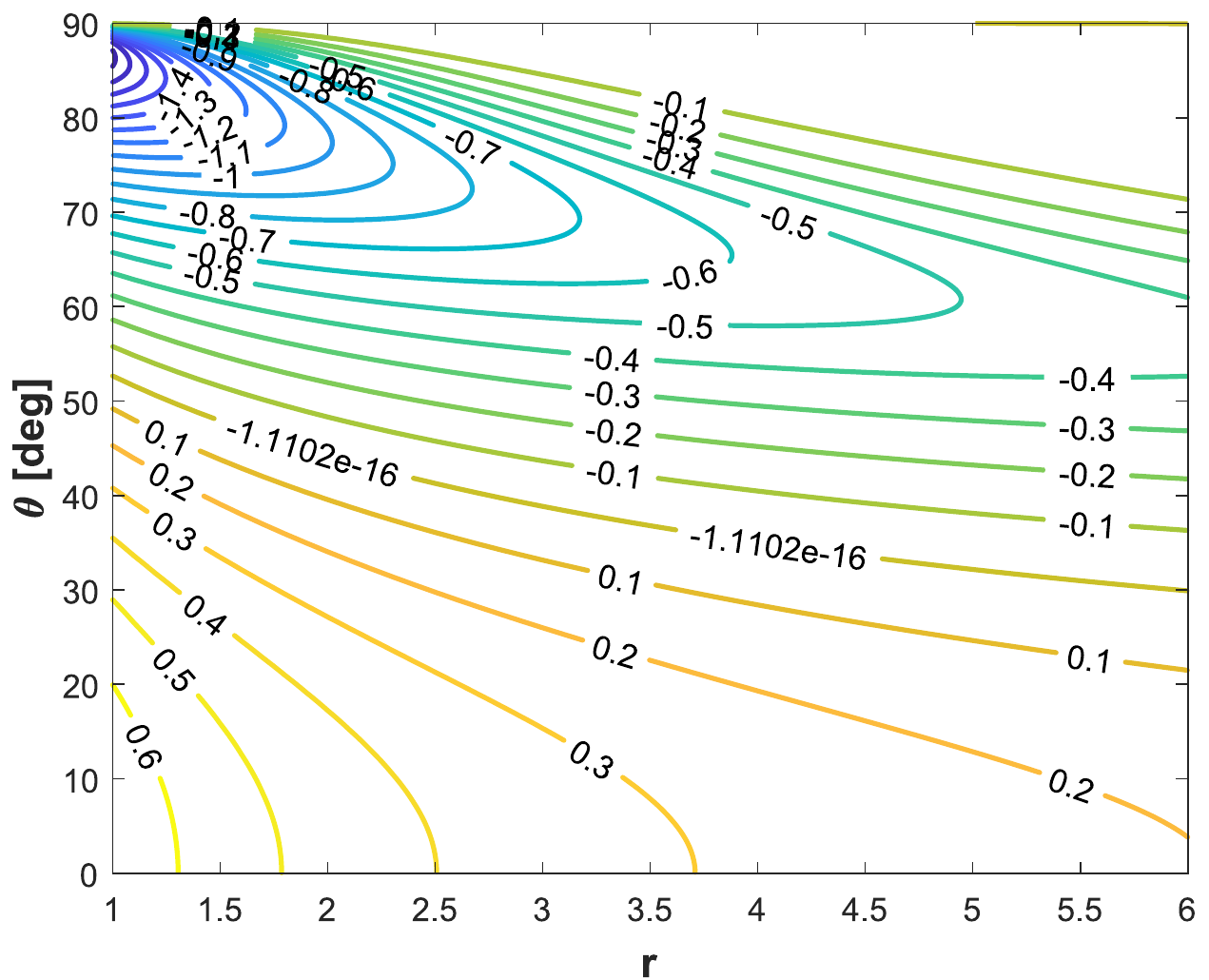}
      \caption{Level curves of Eq. \eqref{eq:levcurv} for a value of $\Gamma^{'} = 10$. }\label{fig:curv1}
    \endminipage\hfill
    \minipage{0.48\textwidth}%
      \centering
      \includegraphics[scale=0.57]{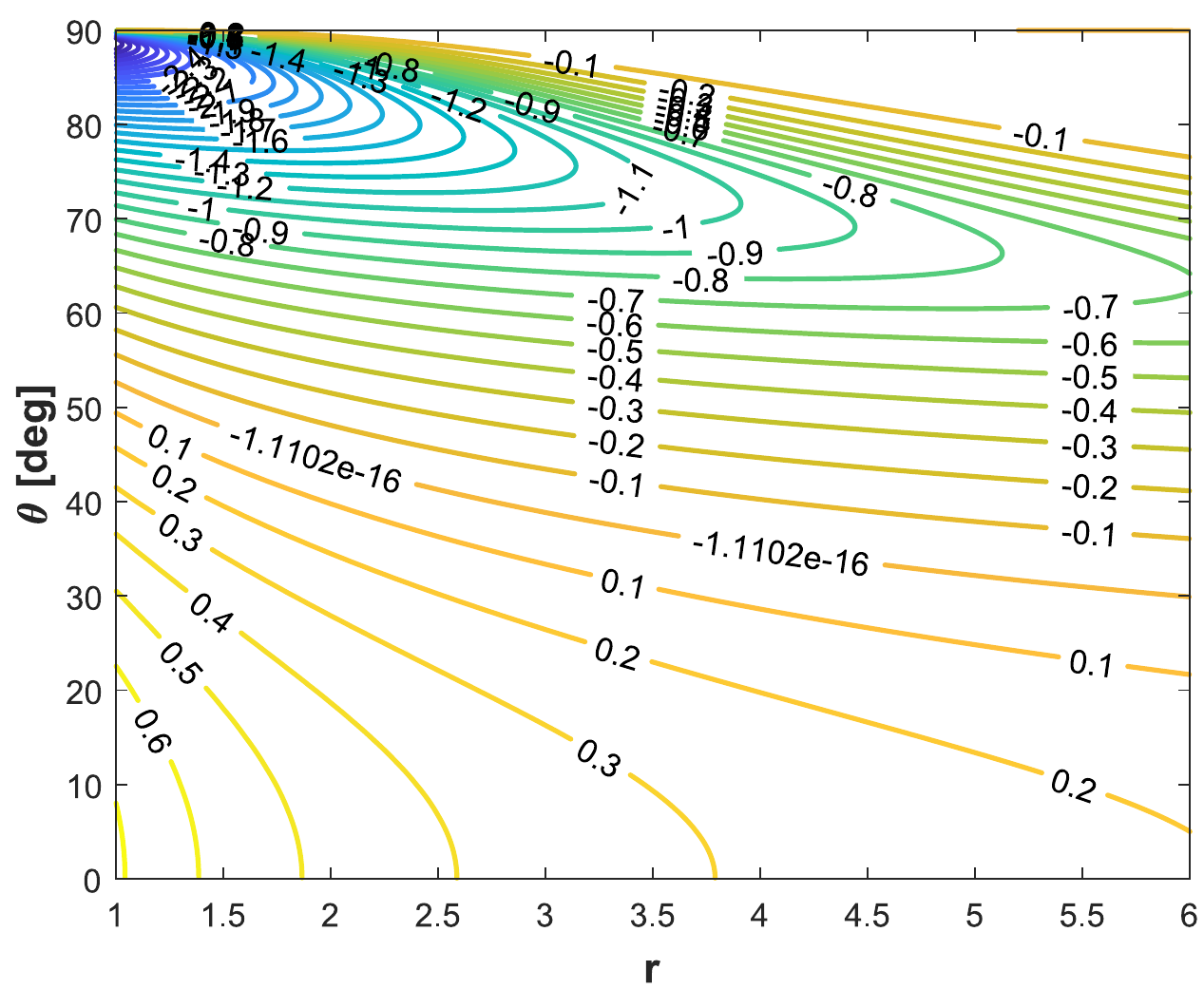}
      \caption{Level curves of Eq. \eqref{eq:levcurv} for a value of $\Gamma^{'} = 50$.}\label{fig:curv2}
    \endminipage
\end{figure}

Once the solutions to  (O-X1), (O-T3) and (O-R2) are obtained, we iterate among them as presented in Algorithm 1 for fixed values of the powers $\P$.

\begin{figure}[!t]
 \removelatexerror
  \begin{algorithm}[H]
   \caption{Optimal TDMA scheduling, UAV trajectories, and beamwidths with fixed power. }
   Set inner iteration number, $j = 1$. \\
   \While { $\frac{ | \mu^j - \mu^{j+1} | }{\mu^j} > \epsilon$ or $j < I_{max,AQR}$ }
   {
      (i) Fix: $\Q^j$ and $\R^j$. Solve (O-X1)  $\xrightarrow{} \ls \A^{j+1}$ \\ 
      (ii) Fix: $\A^{j+1}$ and $\R^j$. Solve (O-T3) $\xrightarrow{} \ls \Q^{j+1}$ \\
      (iii) Fix: $\A^{j+1}$ and $\Q^{j+1}$. Solve (O-R2) $\xrightarrow{} \ls \R^{j+1}$ \\
      $j \xleftarrow{} j + 1$
   }
  \end{algorithm}
  \label{alg:AQR}
\end{figure}

\subsection{Power Allocation Sub-Problem}\label{subsect:power}
Finally, in this section, we consider the power allocation  sub-problem, where the transmit power of both MBS and R-UAVs is jointly optimized assuming fixed values of the TDMA scheduling, UAV trajectories, and antenna beamwidths. The optimization problem, named (O-P1), is given by:
\begin{equation}\label{opt:PowerOriginal}
    \begin{aligned}
    & \underset{\mu, \P}{\text{max}}
    & & \mu \\
    & \text{s.t.} & &  \Rk \geq \mu \\
    & & &\eqref{ct:powerAvgBS}- \eqref{ct:powerR}, \eqref{ct:Causality}
    \end{aligned}
\end{equation}
Due to the non-convexity of Constraint $\eqref{ct:Causality}$, we need to manipulate \eqref{opt:PowerOriginal} to obtain a more tractable problem. To make equations more manageable, we define the following constant terms:
\begin{align*}
    \small{\Gamma_{B,m}^{''}[n] = \frac{ \rhobs{} (r_B[n] + 1)| \hb{} - \hr{} \ls |^{r_{B}[n]}}{ \sigma^2 (||\qbs{} - \qn||^2 + (  \hb{} - \hr{})^{2} )^{ \frac{\kappa + r_B[n]}{2}  } }, \ls \ls \Gamma_{m,k}^{''}[n] =\frac{ \rhoo{}(r_m[n]+1)\hr{}^{r_{m}[n]}}{ \sigma^2 (||\qn - \wk||^2 + \hr{}^{2} )^{ \frac{\kappa + r_m[n]}{2}  }}.}
\end{align*}
In the subsequent, by introducing the set of  slack variables $ \boldsymbol{T} = \{ \tkm{}[n] \ls \ls n = D + 1,\dots, N \ls \ls \forall k,m  \}$, we re-write the original problem, (O-P1), into an equivalent problem, named  (O-P2):
\begin{equation*}
    \begin{aligned}
    & \underset{\mu, \P, \boldsymbol{T}}{\text{max}}
    & & \mu \\
    & \text{s.t.} & &  \sum \limits_{n = D+1}^N \sum \limits_{m = 1}^M \akm \tkm{}[n] \geq \mu \ls \ls \forall k  \\
    & & &\sum \limits_{i = 1}^{n-D} \beta_{B,m}[i] R_{B,m}[i] \geq \sum \limits_{i = D+1}^{n} \sum \limits_{k = 1}^{K} a_{m,k}[i] \tkm{}[i] \mspace{20mu} n = D+1,\dots,N \ls , \ls \forall m\\
    & & & \tkm{}[n] \leq \log_2 ( 1 + \pm{} \Gamma_{m,k}^{''}[n] ) \mspace{20mu} n = D+1,\dots,N \ls , \ls \forall k,m \\
    & & & \eqref{ct:powerAvgBS}- \eqref{ct:powerR}
    \end{aligned}
\end{equation*}

\begin{lemma}\label{Lemma2}
(O-P2) is equivalent to (O-P1).
\end{lemma}
\begin{proof}
The proof can be found in Appendix  \ref{App:ProofLema1}.
\end{proof}

Since (O-P2) is jointly convex with respect to $\mu, \P$ and $\boldsymbol{T}$, we can derive analytical expressions for the powers. As both powers are related via the causality constraint, we split the problem into two sub-problems, named: (i) Optimal R-UAVs to GUs power with fixed MBS to R-UAVs power  and (ii) Optimal MBS to R-UAVs power with fixed R-UAVs to GUs power. Note that directly solving (O-P2) via  solvers, as done in the literature, results in sending more power than what is needed from the MBS to meet the causality constraint. As a result, the two hops would be unbalanced and the sum rate of the first hop would be much higher than that of the second hop. Therefore, our  solution for such a problem  provides a more efficient use of power resources.  

\subsubsection{Optimal R-UAVs to GUs Power Allocation}
The first sub-problem we aim to solve relates the power from the R-UAVs to the GUs, named (O-P2.1), which can be formulated as:

\begin{equation}\label{opt:Power2}
    \begin{aligned}
    & \underset{ \scriptsize{ \pm{}, \tkm{}[n]}  }{\text{max}}
    & &  \mu \\
    & \text{s.t.} & &  \sum \limits_{n = D+1}^N \sum \limits_{m = 1}^M \akm \tkm{}[n] \geq \mu \ls \ls \forall k  \\
    & & &\sum \limits_{i = 1}^{n-D} \beta_{B,m}[i] R_{B,m}[i] \geq \sum \limits_{i = D+1}^{n} \sum \limits_{k = 1}^{K} a_{m,k}[i] \tkm{}[i] \mspace{20mu} n = D+1,\dots,N \ls , \ls \forall m\\
    &  & &  \tkm{}[n] \leq \log_2 ( 1 + \pm{} \Gamma_{m,k}^{''}[n] ) \mspace{20mu}  n = D+1,\dots,N \ls , \ls \forall k,m \\
    & & & \eqref{ct:powerAvgR}, \eqref{ct:powerR}
    \end{aligned}
\end{equation}
After applying the Lagrangian method, the optimal solution to  (O-P2.1) is given by:
\begin{align}\label{eq:optimalpower2}
p_{m}^*[n] = \left\{ \begin{array}{cc} 
                0 & \hspace{5mm} n=1,\dots,D \\
                \Big[  \zeta - \frac{1}{\Gamma_{m,k}^{''}[n]}  \Big]^+ & \hspace{5mm} n = D+1,\dots,N, \\
                \end{array} \right.
\end{align}
where $\zeta$  depends on the Lagrangian multipliers of the problem and the operator $[ x ]^+ = \max(x , 0) $. 

\begin{proof}
The proof can be found in Appendix \ref{App:P1}.
\end{proof}

\subsubsection{Optimal MBS to R-UAVs Power Allocation}
Next, we aim to find the minimum MBS transmit power that satisfies all constraints. Therefore, the next problem, named (O-P2.2), is defined as:
\begin{equation}\label{opt:Power1}
    \begin{aligned}
    & \underset{ \pb{}  }{\text{min}}
    & &   
        \sum \limits_{n = 1}^{N-D} \pb{}  \\
    & \text{s.t.} & & \sum \limits_{i = 1}^{n-D} \beta_{B,m}[i] R_{B,m}[i] \geq \sum \limits_{i = D+1}^{n} \sum \limits_{k = 1}^{K} a_{m,k}[i] R_{m,k}[i], \mspace{30mu} n = D+1,\dots,N \ls , \ls \forall m\\
    & & & \eqref{ct:powerAvgBS}, \eqref{ct:powerBS}
    \end{aligned}
\end{equation}
Similar to (O-P2.1), using the Lagrangian method, the optimal power from the MBS to the $m$-th R-UAV is given by:

\begin{align}\label{eq:optimalpower1}
p_{B,m}^*[n] = \left\{ \begin{array}{cc} 
                \Big[ \xi - \frac{1}{\Gamma_{B,m}^{''}[n]}  \Big]^+ & \hspace{5mm} n=1,\dots,N-D \\
                0 & \hspace{5mm} n = N-D + 1,\dots,N \\
                \end{array} \right..
\end{align}

\begin{proof}
The proof can be found in Appendix \ref{App:P2}.
\end{proof}

Both sub-problems provide  water-filling solutions, taking into account the equivalent channels given by $\Gamma_{m,k}^{''}[n]$ and $\Gamma_{B,m}^{''}[n]$, respectively. Thus, in essence, at each time slot, the MBS and R-UAVs allocate power based on the inverse of the corresponding channel gains. After obtaining the solution to both sub-problems, we iterate until convergence, as summarized in  Algorithm 2.

\begin{figure}[!t]
 \removelatexerror
  \begin{algorithm}[H]
   \caption{Optimal Power Allocation with fixed TDMA scheduling, UAV trajectories, and beamwidths. }
   Set inner iteration number, $i = 1$. \\
   \While { $\frac{ | \mu^i - \mu^{i+1} | }{\mu^i} > \epsilon$ or $i < I_{max,p}$ }
   {
      (i) Given $p_B^{i}[n]$,  solve (O-P2.1)  $\xrightarrow{} \ls p_m^{i+1}[n]$ \\
      (ii) Given $p_m^{i+1}[n]$,  solve (O-P2.2)  $\xrightarrow{} \ls p_B^{i+1}[n]$ \\
      $i \xleftarrow{} i+ 1$
   }
  \end{algorithm}
  \label{alg:P1}
\end{figure}

\subsection{Convergence}
Based on the solutions to the previous sub-problems, we propose an iterative method for the initial non-convex problem in which we optimize four sets of variables: TDMA scheduling,   UAV trajectories, directivity degrees, and power allocation. Since the four sub-problems are convex optimization problems, they can be solved using existing polynomial-time algorithms.  In addition, in Algorithm 3, we iterate between power optimization and Algorithm 1 because including the power optimization into Algorithm 1 generally converges to lower values of the cost function.
The convergence of the proposed BCD algorithm in Algorithm 3 is guaranteed by the following proposition:

\begin{prop}
The sequence of objective values generated by the proposed BCD approach in Algorithm 3 is monotonically non-decreasing with an upper bound, and therefore converges.
\end{prop}
\begin{proof}
 The proof can be found in Appendix \ref{App:ProofProp1}.
\end{proof}

\begin{figure}[!t]
 \removelatexerror
  \begin{algorithm}[H]
   \caption{Algorithm for the TDMA Scheduling, UAV Trajectories, Directivity Degree, and Power Allocation.}
   Set outer iteration, $k = 1$. \\
   Initialize: $\P^{1}$, $\Q^{1}$ and $\R^{1}$. \\
   \While { $\frac{ | \mu^k - \mu^{k+1} | }{\mu^k} > \epsilon$ or $k < I_{max}$ }
   {
      (i) Fix: $\P^k$ and run Algorithm 1 to obtain $\A^{k+1}, \Q^{k+1}, \ls \R^{k+1} $. \\
      (ii) Fix: $\A^{k+1}, \Q^{k+1}, \ls \R^{k+1} $ and run Algorithm 2 to obtain $\P^{k+1}$.\\
      $k \xleftarrow{} k + 1$
   }
  \end{algorithm}
  \label{alg:Al1}
\end{figure}

\subsection{Algorithm Initialization}
In this subsection, we explain the initialization methods for the UAV trajectories. We distinguish between the MBS and R-UAVs since $\hb{} > \hr$ to avoid any possible collision. 

\paragraph{MBS Initialization}
    As the MBS is assumed to transmit data to all R-UAVs, a natural initialization is to place the MBS in the middle of all GUs, which consequently will be in the middle of the R-UAVs as well. Therefore, we first find the mass center of the GUs,  $\boldsymbol{c}_{K} = \frac{1}{K} \sum_{\forall k} \wk{}$, and  create a circular trajectory around it. To calculate the radius of the circle, we take into account the maximum velocity between two generic consecutive snapshots $n$ and $n+1$. The minimum distance between two points in a uniformly sampled circle is  $2R\sin(\frac{\pi}{N})$, where $R$ is the radius and $N$ is the number of points on the circle. To guarantee that the UAV can fly such a distance, at each step, we need $2R\sin(\frac{\pi}{N}) \leq (V_{B}\delta) $ which results in the maximum radius of $R_{max} = \frac{(V_{B}\delta)}{2} (\sin(\frac{\pi}{N}))^{-1}$. Therefore, the initial trajectory for the MBS is:
        \begin{align*}
            \boldsymbol{q}_{B}^1[n] = \boldsymbol{c}_{K} +  \gamma R_{max} \left[ \cos\left(\frac{2\pi}{N-1}n\right) \ls \ls  \sin\left(\frac{2\pi}{N-1}n\right) \right] \ls \ls \ls \forall n,
        \end{align*}
        where $0 \leq \gamma \leq \gamma_{max} $ ensures the movement-related power constraints are met in the first inner iteration of Algorithm 1, $j = 1$. Note that a circle of radius $R_{max}$ may not satisfy such constraints. The value of $\gamma_{max}$ can be computed from Eq. \eqref{eq:consumption}, assuming constant velocity, and the power of the batteries, $P_{UAV}$.
 
    \paragraph{R-UAVs Initial Trajectory}
    For the R-UAVs, we combine the Circle Packing (CP) technique \cite{CPTech}  with the movement-related power constraints. After finding the mass center, $\boldsymbol{c}_{K}$, we compute the minimum radius circle that contains all GUs by $r_K = \max ||\boldsymbol{c}_{K} - \wk{}|| $. Given $M$ R-UAVs and $r_K$, by applying the CP technique we obtain $M$ centroids, denoted by $\boldsymbol{c}_m$, and its respective radius $r_m^{CP}$. To have, on average, the same number of users inside and outside $r_m^{CP}$, we scale $r_m^{CP}$ by a factor of $0.5$. To take into account the maximum velocity constraint, we calculate $R_{max}$ using the same approach described for the MBS Initialization. Therefore, the radius for the circular trajectory of  R-UAVs is given by $r_m = \min ( R_{max}, \frac{r_m^{CP}}{2}) $. As a consequence, the initial trajectories for the R-UAVs are given by:
        \begin{align*}
            \boldsymbol{q}_{m}^1[n] = \boldsymbol{c}_{m} +  \gamma r_m \left[ \cos\left(\frac{2\pi}{N-1}n\right) \ls \ls  \sin\left(\frac{2\pi}{N-1}n\right) \right] \ls \ls \ls \forall n,
        \end{align*}
    where $\gamma$, again, ensures the trajectory-related power constraint is met.

\section{Simulation Results}\label{Sec:Res}
In what follows, we present simulation results for the proposed multi-UAV relay network. 
In our simulations, we generate $K = 7$ users  randomly in an $800 \times 800 \ls m^2$ area and set $\kappa = 2$, $d_{0,A}^{\kappa} = - 35$ dB, $d_{0,G}^{\kappa} = - 50$ dB, and  $\sigma^2 = - 100$ dBm, as suggested in the literature. In addition,  $P_{B,avg} = P_{R,avg} = 2$mW with a peak power of $P_{B,max} = P_{R,max} = 4$mW and $A = 1$, which are common values used in the literature as well. The altitudes of the MBS and R-UAVs are fixed at  $H_B = 200$m and $H_R = 100$m, respectively, being in concordance with the LoS channel assumption. The maximum UAV velocity is set to $V_{B} = V_{R} =  50$ m/s with a minimum safety distance $d_{min} = 10$m. The minimum and maximum beamwidth degrees for both the MBS and R-UAVs are $r_{min} = 1$ and $r_{max} = 6$, respectively. In addition, unless specified, the processing delay is set to $D = 1$ slot and the on board power is $P_{UAV} = 6.5$ kW.
\par

We first illustrate the convergence of Algorithm 1 and the gains provided by the beamwidth optimization in Fig. \ref{fig:convergence}. To simplify the presentation, we consider $M = 1$ in this case. We evaluate two scenarios, for $T = 50$s and  $T = 60$s. In both cases, we include two curves. The blue curves (solid and dashed) are the results of optimizing $\A$ and $\Q$ with fixed  beamwidths $\boldsymbol{R}$, $r_{B}[n] = r_m[n] = 2$. If we add the beamwidth sub-problem, as in Algorithm 1, we obtain the red curves  for $T = 50$s (dashed) and $T = 60$s (solid). For the case where $T = 50$s, the gain due to the adaptive beamwidths is $ \frac{ \mu^{*}(\boldsymbol{X}, \boldsymbol{Q}, \boldsymbol{R}) }{\mu^{*}(\boldsymbol{X}, \boldsymbol{Q})} \approx 1.88$, while for $T = 60$s, it is even greater, i.e., $ \frac{ \mu^{*}(\boldsymbol{X}, \boldsymbol{Q}, \boldsymbol{R}) }{\mu^{*}(\boldsymbol{X}, \boldsymbol{Q})} \approx 1.98$. The higher $T$, the more R-UAVs can fly near the GUs and therefore use more directive patterns, which yields to an improvement of the end-rate. Finally, running Algorithm 3 for $T = 50$ provides a minimum achieved rate of $0.0320$ bps/Hz, where the combined gain provided by the optimization of $\R$ and $\P$ is $3.18$. Similarly, implementing Algorithm 3 for $T = 60$s provides a combined gain of $3.29$, where the minimum achieved rate is $0.04$ bps/Hz. \par 

In Fig. \ref{fig:convergence2}, we include the evolution of Algorithm 3 for two cases, $[T = 60, \ls M = 1]$ and $[T = 30, \ls M = 2]$ (blue curves). Recall that each iteration, $k$, is composed by two inner iterative algorithms. Before reaching $k = 1$, Algorithm 1 converges after $60$ inner iterations, and then, Algorithm 2 produces a steep increase in the rates after optimizing the powers. A similar pattern is seen before reaching $k = 2$, while the third iteration of Algorithm 3 shows that the algorithm has converged. In addition, we include the evolution of the Jain's Fairness Index (F.I.) \cite{Jain}, defined as: $F.I. = \frac{(\sum \limits_{k = 1}^K \Rk{})^2  }{K \sum \limits_{k = 1}^K \Rk{}^2 }$,
by the red curves in Fig. \ref{fig:convergence2}. A complete fairness, i.e. $F.I. = 1$, is achieved very quickly. While the final Jain's Fairness Index for all simulations in this section is $F.I. = 1$, for the sake of brevity, we will not report it in the rest of the paper.

\begin{figure}
    \minipage{0.42\textwidth}
    \centering
      \includegraphics[scale=0.48]{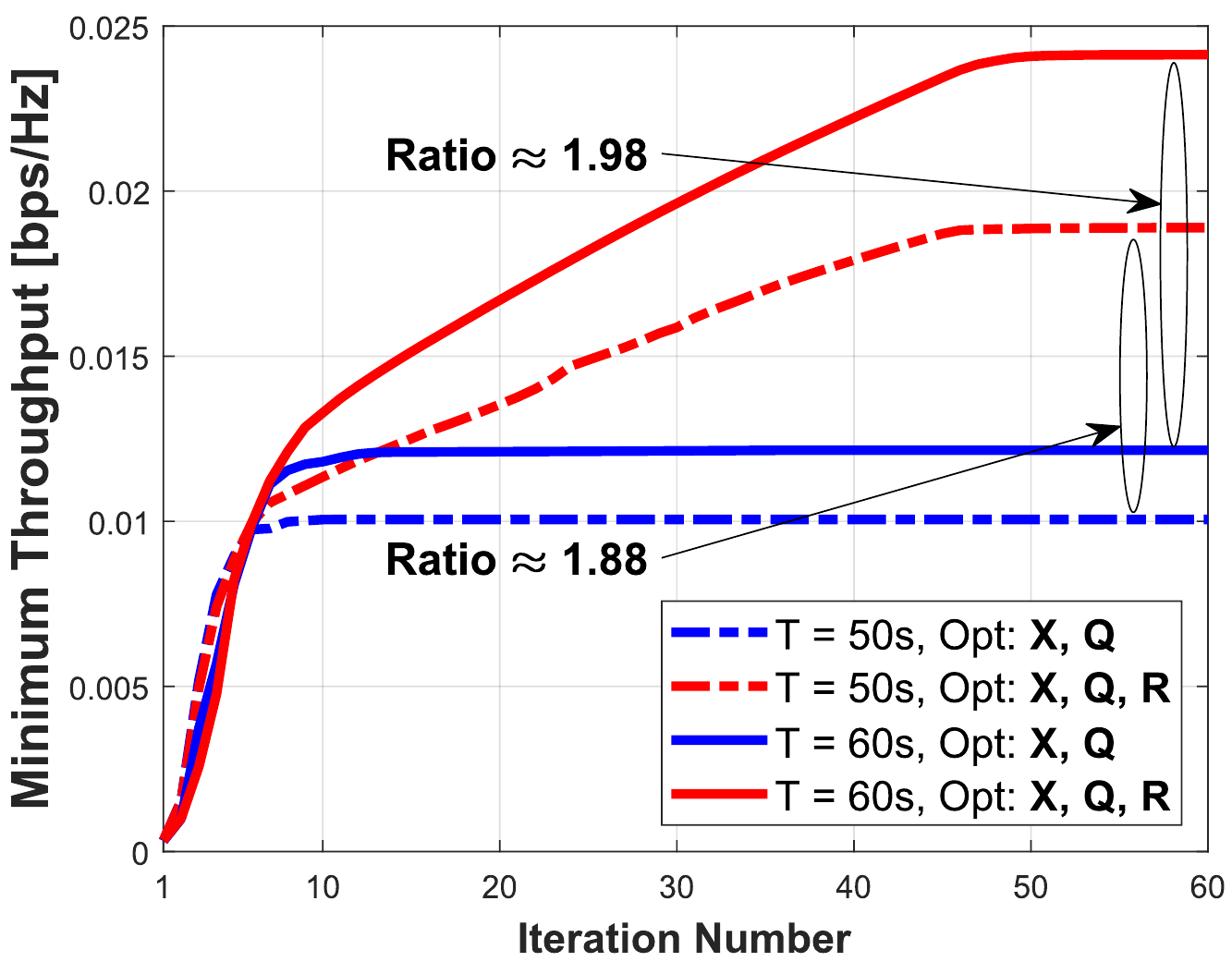}
      \caption{Convergence of Algorithm 1. 
      }\label{fig:convergence}
    \endminipage\hfill
    \minipage{0.46\textwidth}
    \centering
      \includegraphics[scale=0.49]{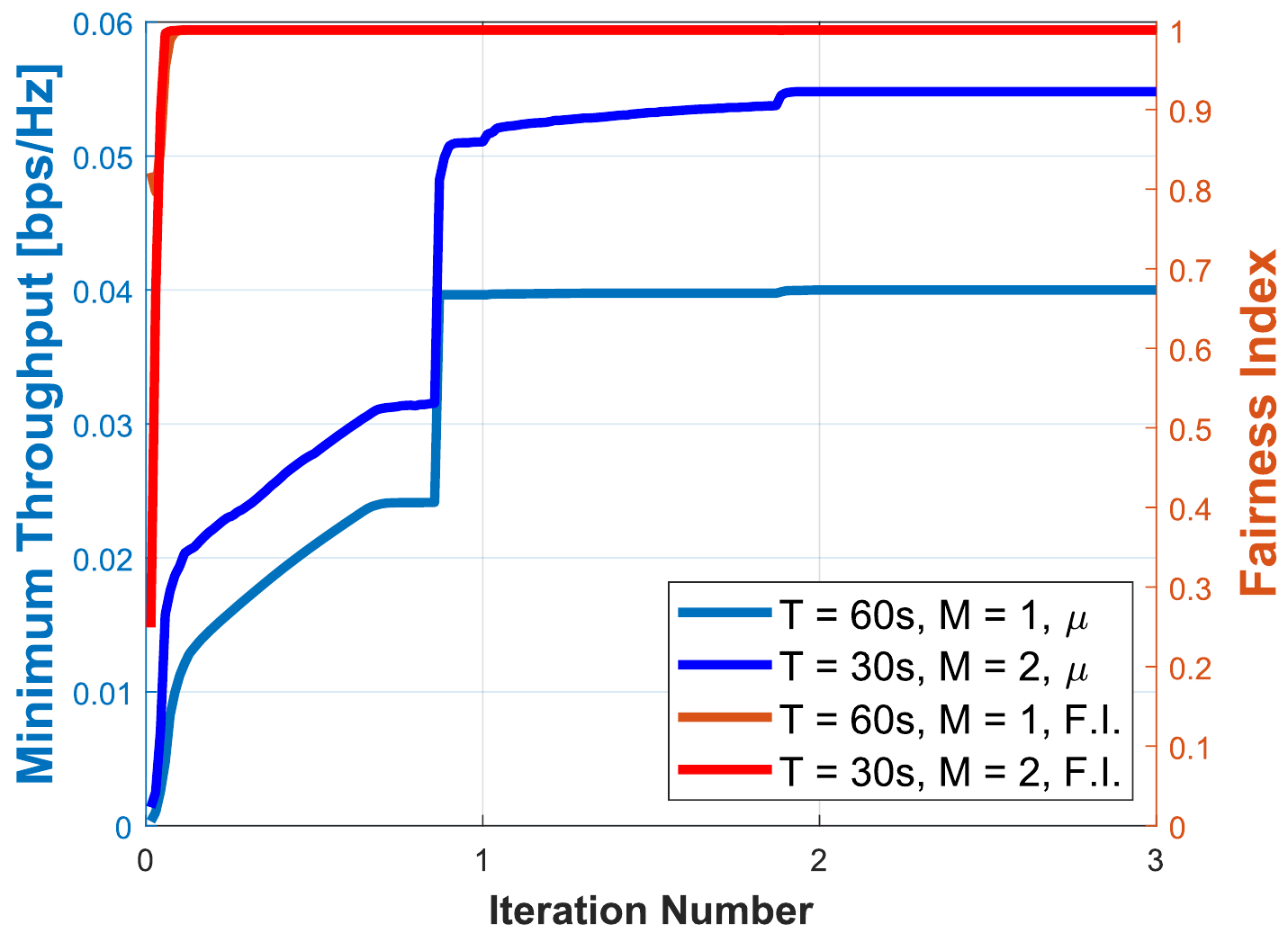}
      \caption{Convergence of Algorithm 3 
      (blue curves) and the evolution of the Jain's Fairness Index (F.I.) (red curves).  }\label{fig:convergence2}
    \endminipage\hfill
\end{figure}

In Fig. \ref{fig:RatevsT}, we present the evolution of the minimum GU rate as a function of the total flying time $T$. To emphasize the gain produced by the optimization of directional antennas, we include the following cases: (i) optimize $\A$, $\Q$ and $\P$ for fixed $\R$, $r_B[n] = r_m[n] = 2$ for $M = 1$ and $M = 2$ (dashed blue and red, respectively) and (ii) optimize $\A$, $\Q$, $\R$ and $\P$ for $M = 1$ and $M = 2$ (solid blue and red, respectively). We also include a static deployment with $M = 2$ relays (solid-black). Clearly, the optimization of the beamwidth degrees provides a huge improvement. The gain is large enough to make it possible to achieve the same minimum throughput with smaller number of relays for large flying times. 
\par
As mentioned in previous sections, we want to study the consequences of different processing delay times at the R-UAVs, given by the variable $D$. In Fig. \ref{fig:delay}, we include the variation of the cost function for different values of $D$. Again, the solid curves correspond to a set-up where we allow adaptive beams, while the dashed curves represent fixed $r_B[n] = r_m[n] = 2$. First, notice how for $M = 1$ and $T = 40$ (solid and dashed red) the impact of the delay is smaller. To the contrary, for $M = 3$ and $T = 30$s (solid and dashed black), the rate greatly decreases as the delay increases. This is because for smaller flying times, UAVs have less slots to relay information as the delay increases. As a result,  for low flying time missions, high delay scenarios do not perform well even if more R-UAVs are used. Nevertheless, the achieved GU rates are improved by adding the optimization of $\R$ into the set-up instead of keeping it fixed at $r_B[n] = r_m[n] = 2$ (solid vs dashed curves).

\begin{figure}
    \minipage{0.45\textwidth}
    \centering
      \includegraphics[width=7cm]{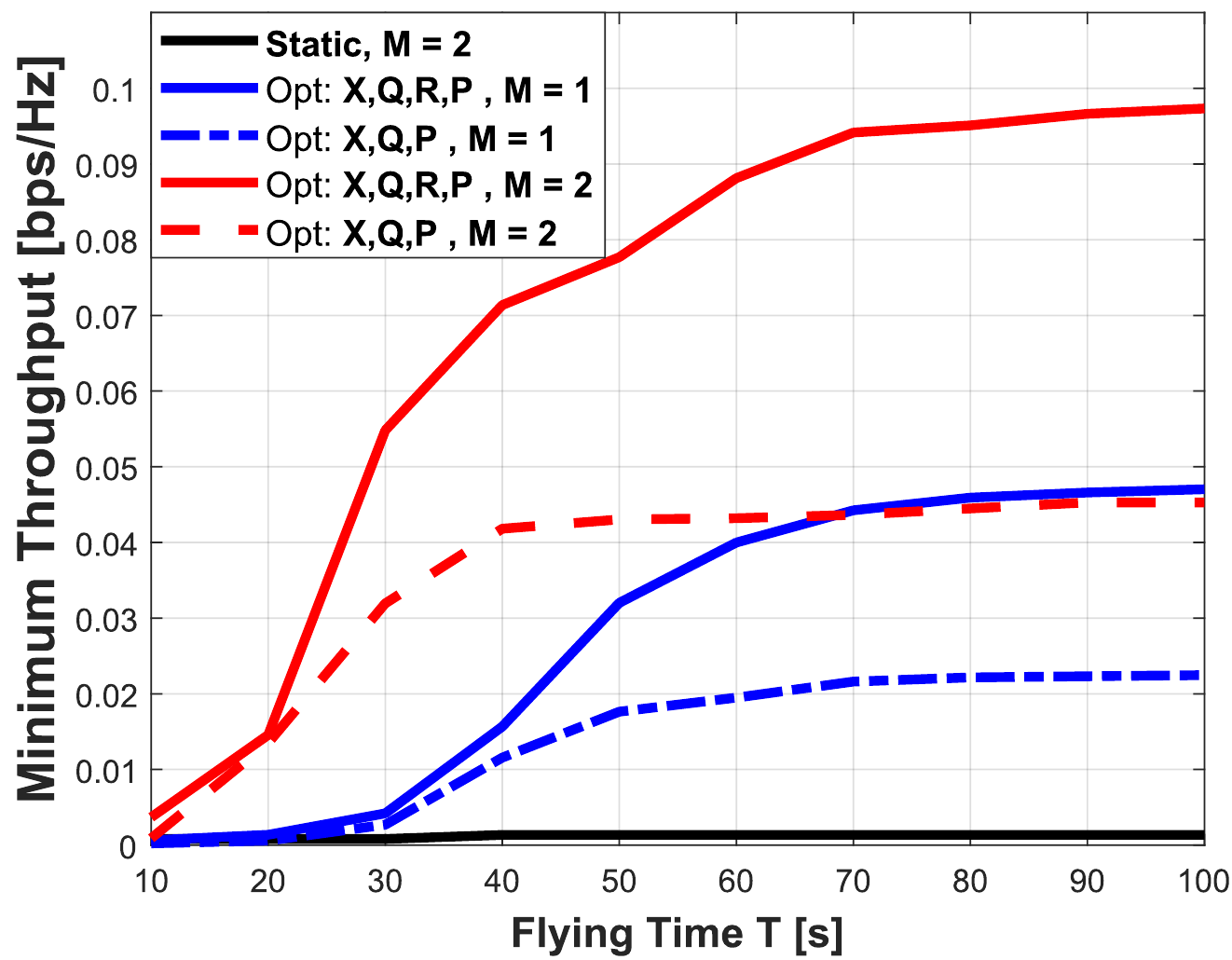}
      \caption{Evolution of the minimum throughput as a function of the flying time.  }\label{fig:RatevsT}
    \endminipage\hfill
    \minipage{0.45\textwidth}
    \centering
      \includegraphics[width=7cm]{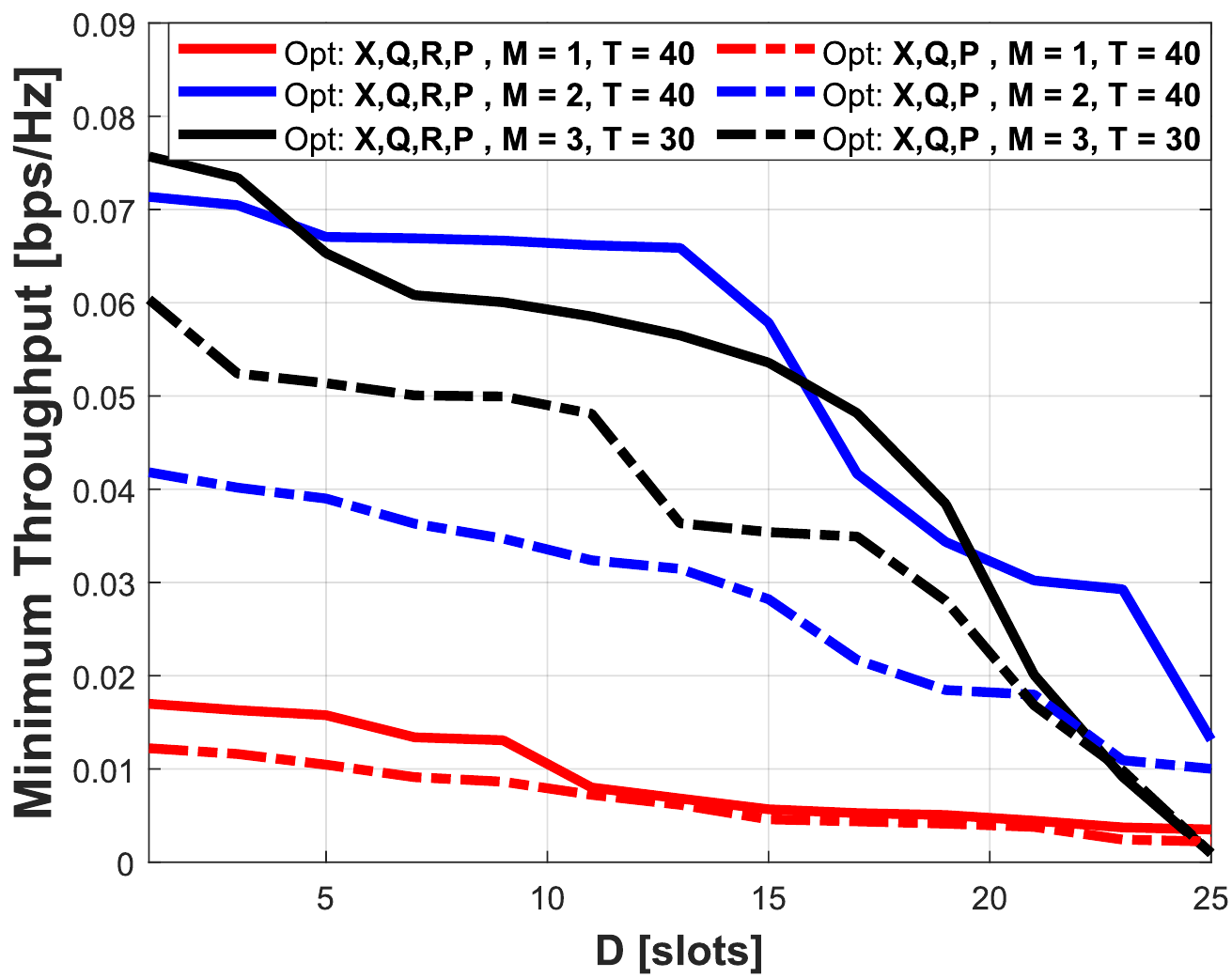}
      \caption{Evolution of the minimum throughput as a function of the delay. }\label{fig:delay}
    \endminipage\hfill
\end{figure}

Fig. \ref{fig:energy} presents the variation of the minimum GU rate as a function of the stored power at the UAV batteries, $P_{UAV}$. Note that we evaluate the same scenarios as in Fig. \ref{fig:delay}. The gap between solid and dashed lines increases as $P_{UAV}$ increases. The reason is mainly due to the fact that increasing $P_{UAV}$ allows R-UAVs to get closer to GUs and therefore exploit the capabilities of having adaptive beamwidths. However, for the cases where the amount of on board power is low, the difference between the the solid and dashed lines decreases, as  UAVs have less  freedom to move. Therefore,  UAVs will tend to use wider beams in such cases, closer to $r_B[n] = r_m[n] = 2$ that represents the dashed curves.

\begin{figure}
    \minipage{0.41\textwidth}
      \includegraphics[width=\linewidth]{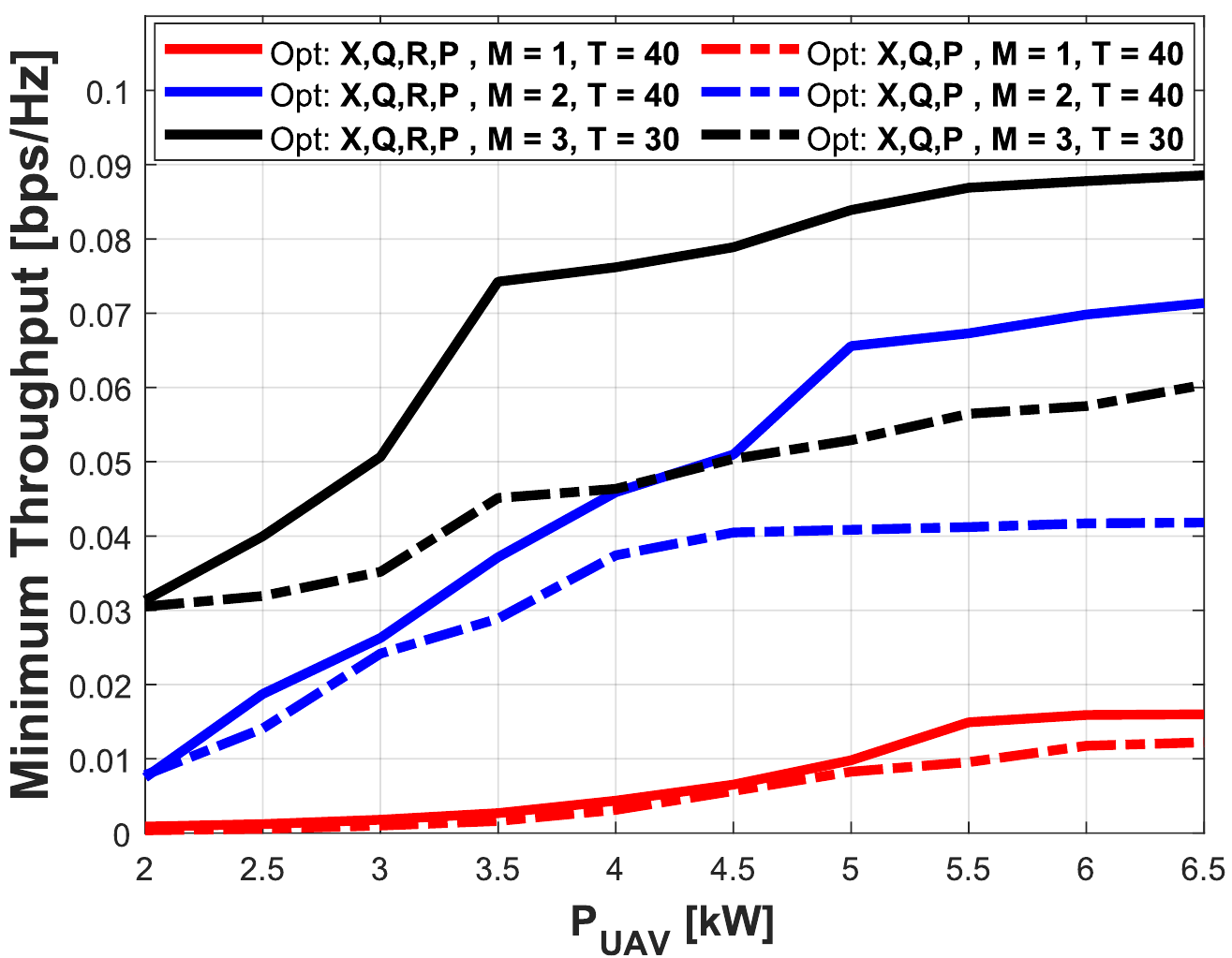}
      \caption{Evolution of the minimum throughput as a function of the on board power, $P_{UAV}$. }\label{fig:energy}
    \endminipage\hfill
    \minipage{0.43\textwidth}
      \includegraphics[width=\linewidth]{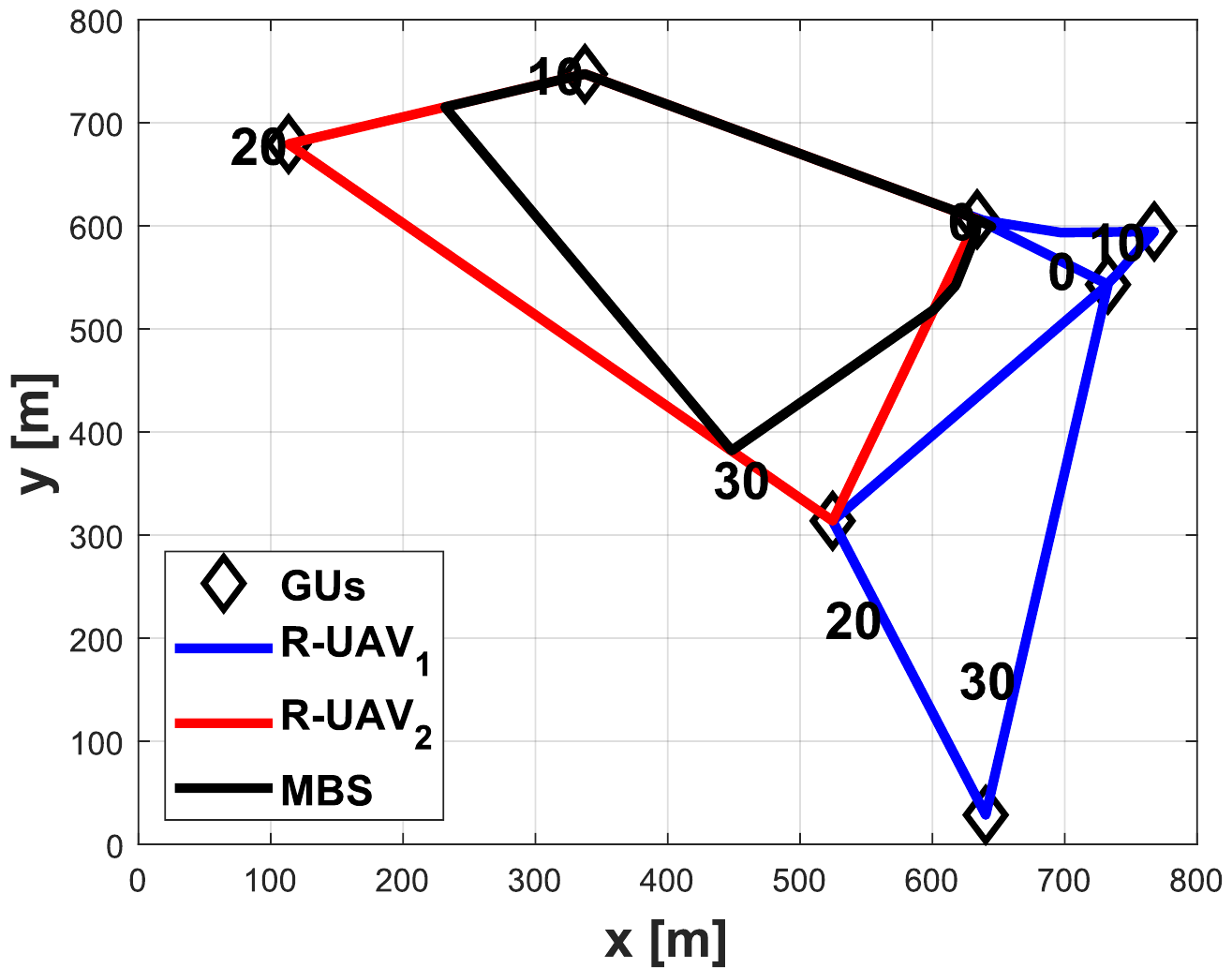}
      \caption{UAV trajectories for $T = 40$s.  }\label{fig:trajectories}
    \endminipage\hfill
\end{figure}

In Fig. \ref{fig:trajectories}, we include UAV trajectories for the flying time $T = 40$s. Both R-UAVs reach all GUs while the MBS stays near both R-UAVs in order to provide them with the data. As a result of trajectory optimization, UAVs avoid an inter-UAV distance smaller than $d_{min}$. 
We present the time index of some trajectory points as a reference to indicate how the R-UAVs coordinately move. In addition, we include the evolution of the MBS transmit power in Fig. \ref{fig:powerbs}. As a result of the optimization, the MBS only transmits at a few time slots, which yields to a more efficient use of the power resources. 
We also include the optimal power allocation (blue) of the two relays in Fig. \ref{fig:trajectories} as a function of the time in Figs. \ref{fig:rp1} and \ref{fig:rp2}. In the same figures, we include the 2D distance from the R-UAVs to the GUs scheduled to receive data at each time (red). Both figures show that the results derived in Section \ref{subsect:power} match with the simulations in the sense that the transmit power depends on the inverse of the channel, being a function of the distance. In fact, when the 2D distance between R-UAV and GU is approximately more than $200$m, they do not transmit power, while for the cases where the R-UAVs fly on top of the GUs, e.g. 2D distance of $0$m, they transmit at a maximum power, $4$mW. 

\begin{figure}
    \minipage{0.285\textwidth}
      \includegraphics[width=\linewidth]{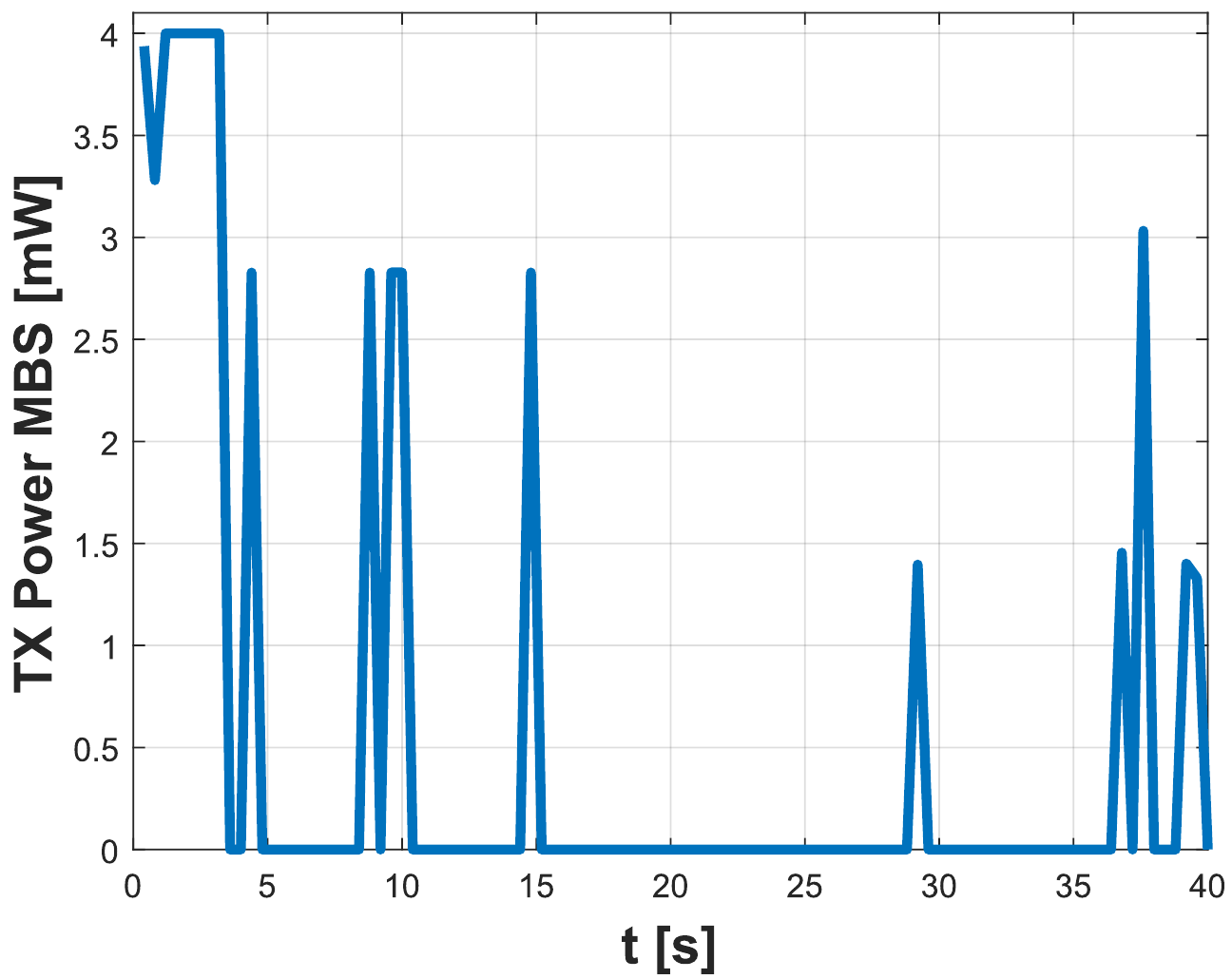}
      \caption{MBS optimal transmit power for $T = 40$s for the scenario presented in Fig. \ref{fig:trajectories}.  }\label{fig:powerbs}
    \endminipage\hfill
    \minipage{0.31\textwidth}
      \includegraphics[width=\linewidth]{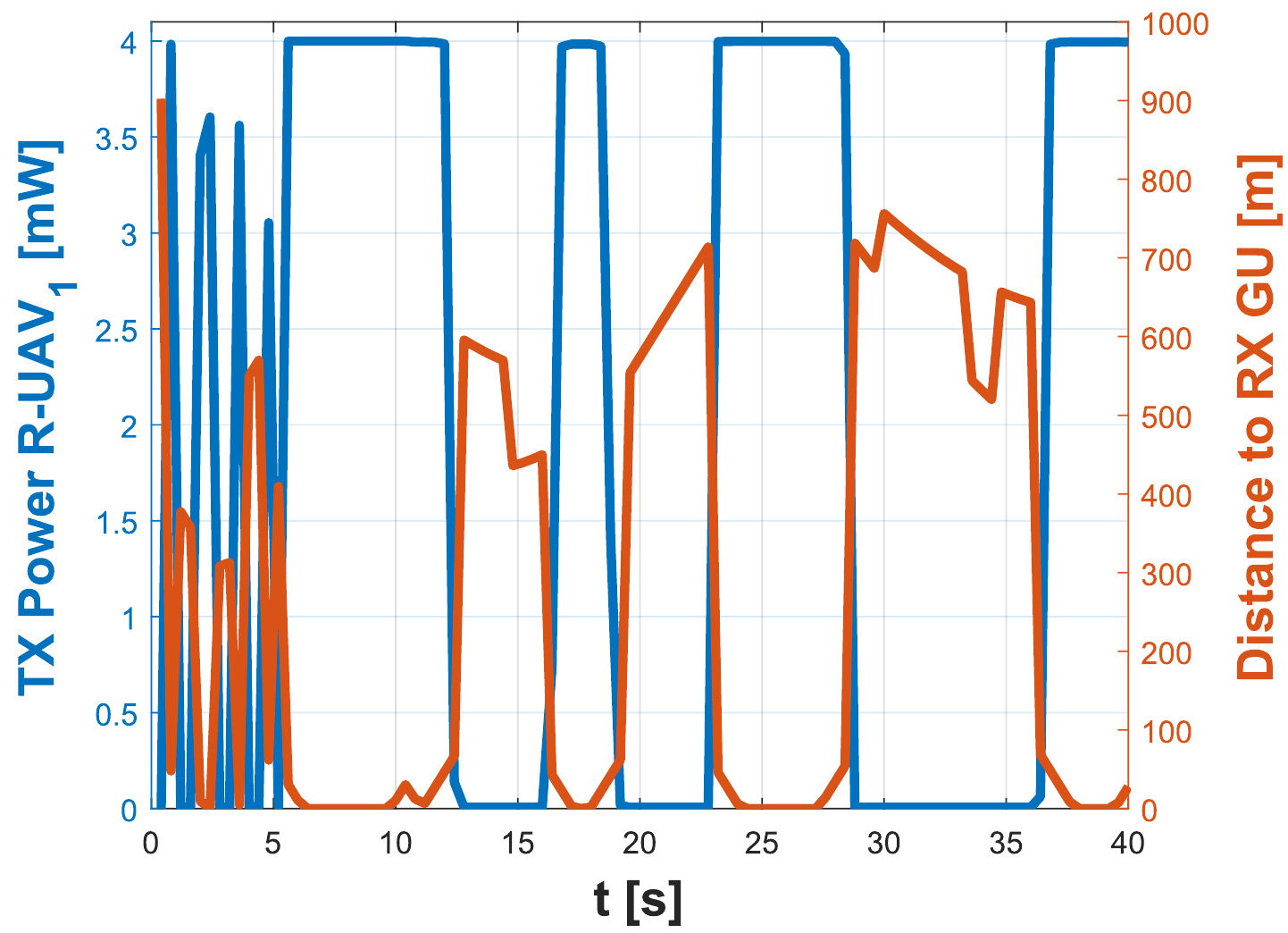}
      \caption{R-UAV$_1$ optimal transmit power (blue) for $T = 40$s  with the  distance to the scheduled GU (red).  }\label{fig:rp1}
    \endminipage\hfill
    \minipage{0.31\textwidth}
      \includegraphics[width=\linewidth]{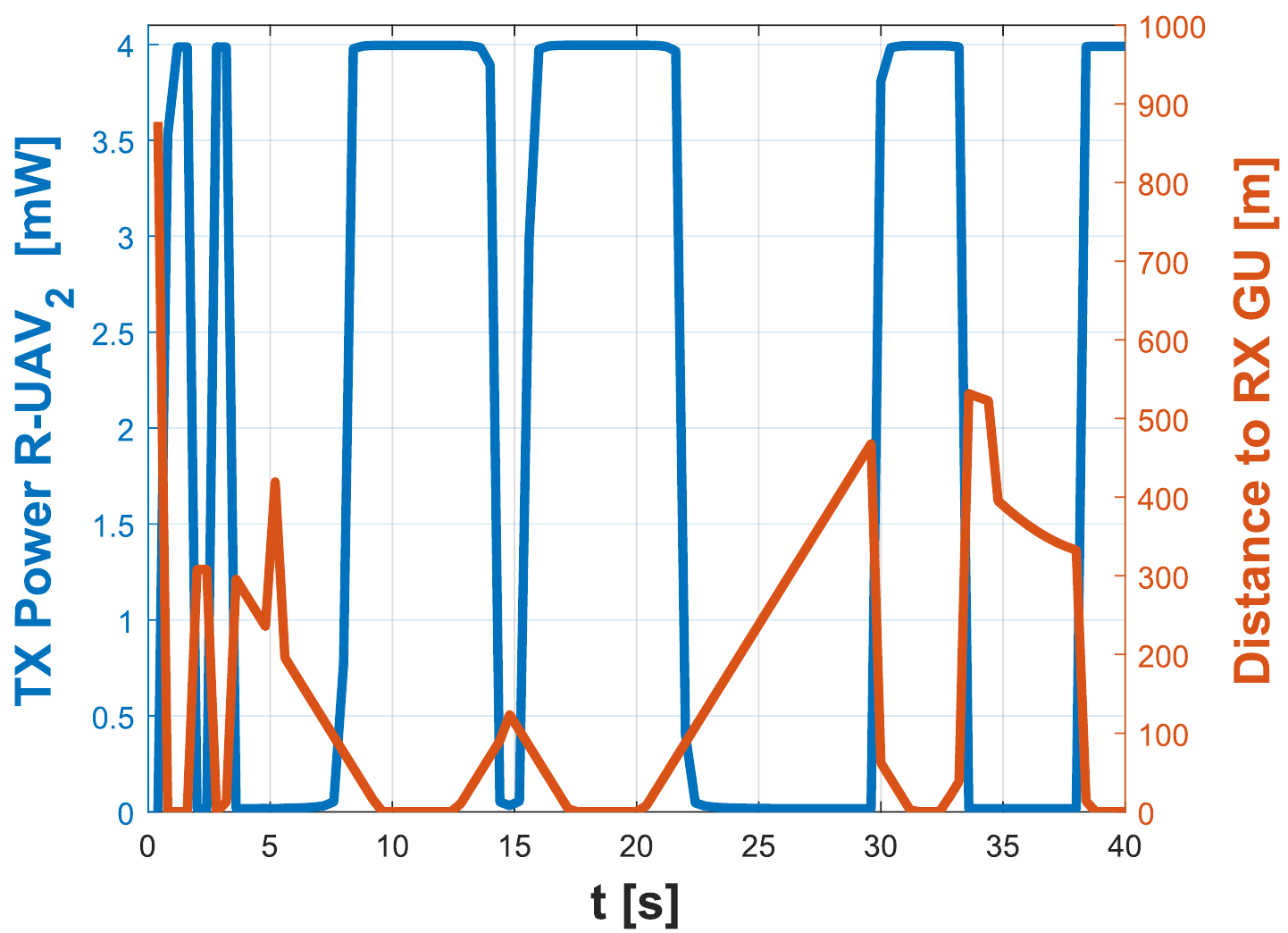}
      \caption{R-UAV$_2$ optimal transmit power (blue) for $T = 40$s with the  distance to the scheduled GU (red).  }\label{fig:rp2}
    \endminipage\hfill
\end{figure}

Finally, for the same scenario presented in Fig. \ref{fig:trajectories}, we include the optimal beamwidth degrees (blue) for the R-UAVs in Figs. \ref{fig:re1} and  \ref{fig:re2}. We also provide the elevation angle between R-UAV and GU scheduled to receive data (red), defined as $\arccos \big(\frac{\hr{}}{\sqrt{ ||\qn{} - \wk{}||^2 + \hr{}^2 }} \big)$. Both figures show that, for low elevation angles, a higher value of $r_m[n]$ is preferred. A low elevation angle means the source flies nearly on top of the receiver, and therefore increases the value of $r_m[n]$ to create a more directive and focused beam. On the contrary, for high elevation angles, e.g. when the source is far from the receiver, a lower value of $r_m[n]$ is preferred. A lower value of $r_m[n]$ creates a wider and less directive beam to cover users at high elevation angles and still provides service. These results are in concordance with the discussion in Section \ref{subsect:beam} where the same conclusions were derived for the special cases that the problem is convex. Similar patterns are observed for the MBS. However, due to lack of space, we do not include them.

\begin{figure}
    \minipage{0.45\textwidth}
      \includegraphics[scale=0.5]{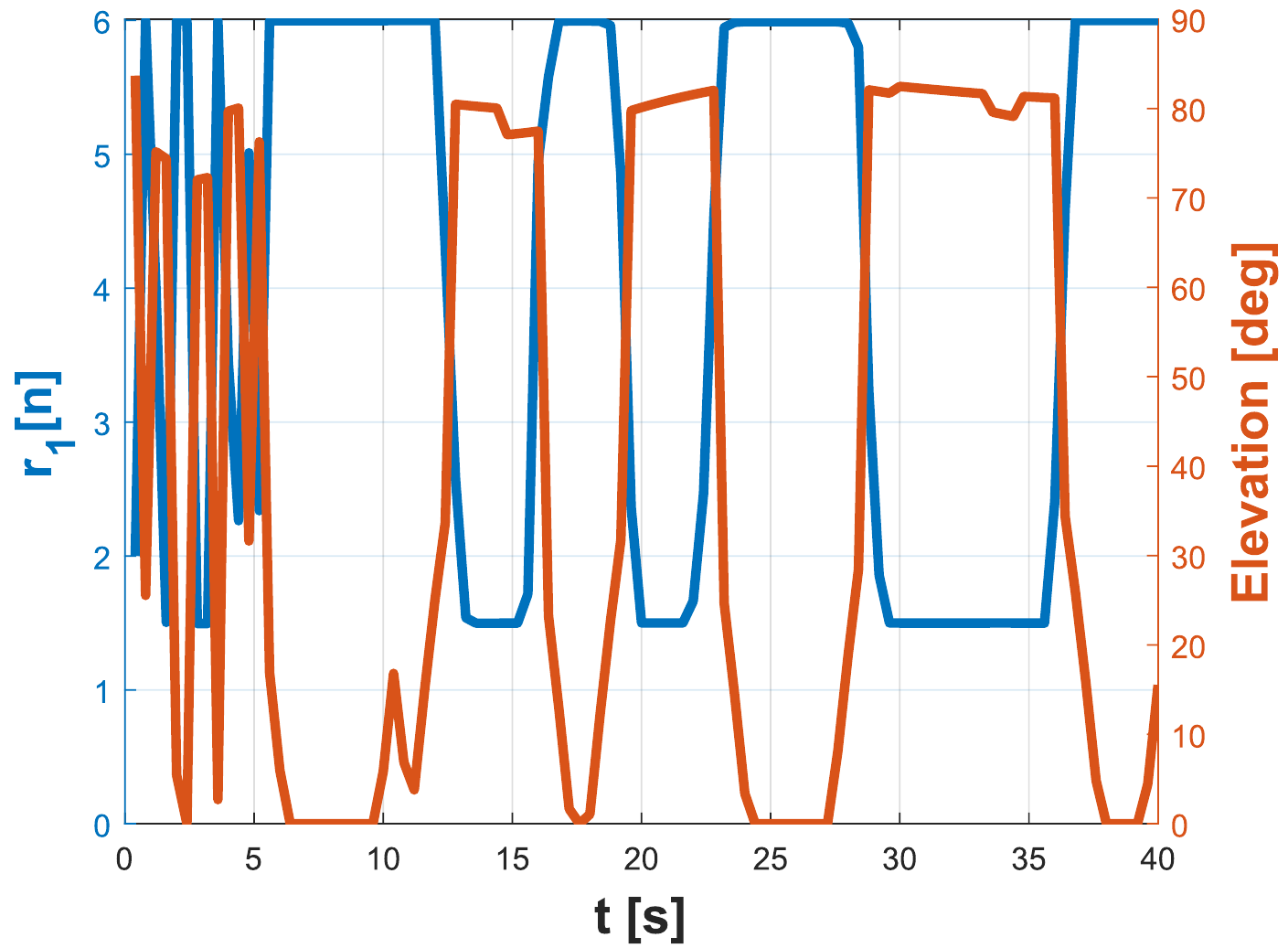}
      \caption{R-UAV$_1$ optimal beamwidths  (blue) for $T = 40$s alongside the  elevation to the scheduled GU  (red). }\label{fig:re1}
    \endminipage\hfill
    \minipage{0.45\textwidth}
      \includegraphics[scale=0.5]{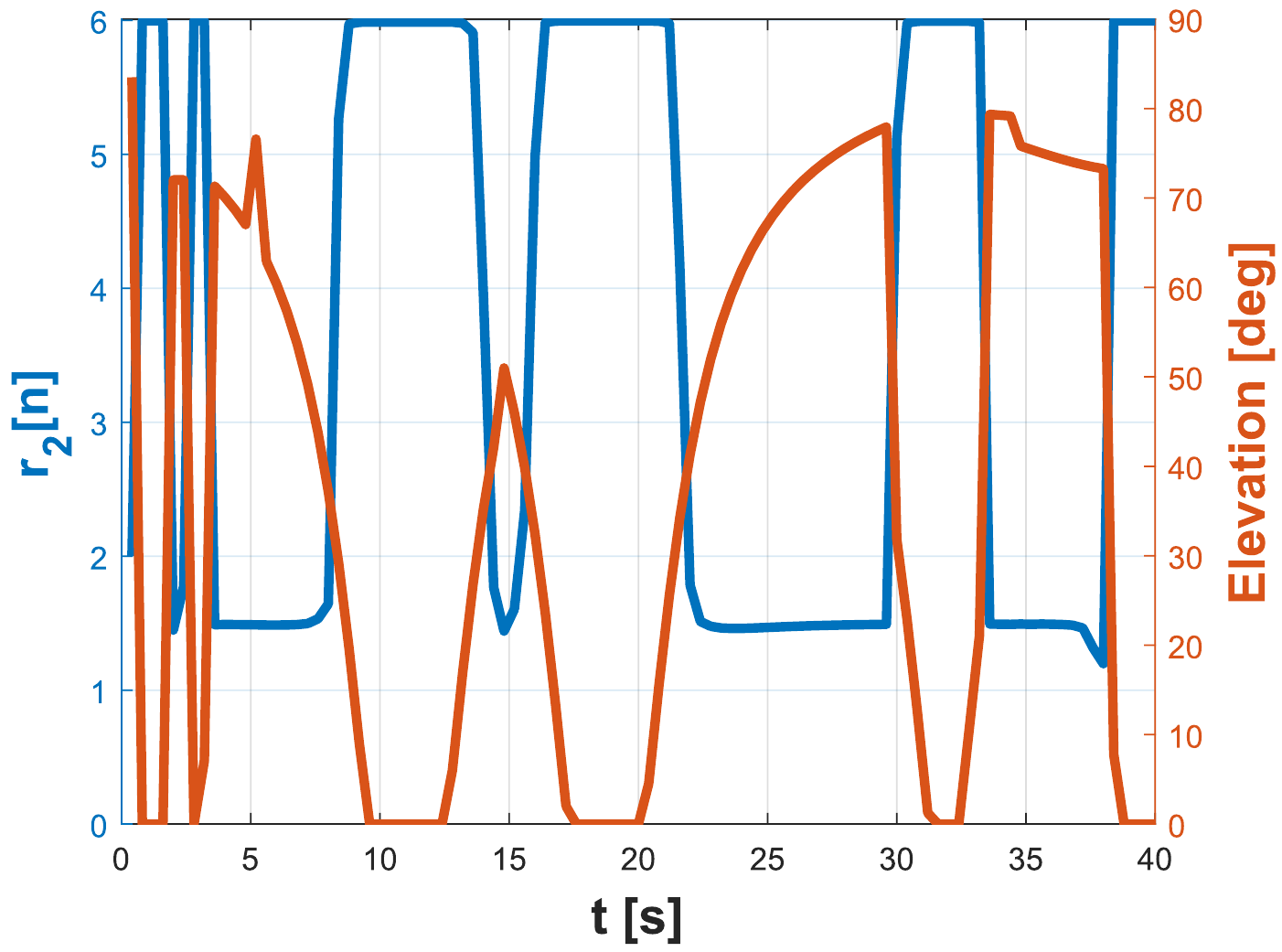}
      \caption{R-UAV$_2$ optimal beamwidths  (blue) for $T = 40$s alongside the  elevation to the scheduled GU (red). }\label{fig:re2}
    \endminipage\hfill
\end{figure}

\section{Conclusion}\label{Sec:Concl}
We studied a down-link two-hop multi-UAV relaying system for the maximize minimum GU throughput problem. We added angle-dependent antenna radiation patterns, producing a more realistic and accurate model. We provided numerical results, showing the intuition behind the beamwidth optimization problem, even if the problem is non-convex. Depending on the elevation between the source and destination, a narrower or wider beam is preferred. In addition, we derived an analytical water-filling type solution for the power allocation problem. We provided numerical results demonstrating the improvement of adding such features in the UAV trajectory problem. For simplicity, we used adaptive beamwidths only for transmission. However, the presented formulation, with minor changes, is applicable to the case with adaptive beamwidths at the receiver side. 
Finally, a very similar formulation works for the up-link  scenario as well as if more hops were used.

\appendices
\section{Proof of Lemmas 1 and 2}\label{App:ProofLema1}
The proof for Lemmas \ref{Lemma1} and \ref{Lemma2} are similar. For the sake of brevity, 
we present the proof for Lemma \ref{Lemma2} and a similar procedure applies to Lemma \ref{Lemma1}.
\begin{proof}
 Assume we have reached the optimal solution to (O-P2). If, for some generic snapshot, $n$, constraint $\tkm{}[n] \leq \log_2 ( 1 + \pm{} \snrr{}[n] )$ is not satisfied with equality, we can reduce the corresponding power $\pm{}$ to satisfy it with equality without a decrease in the objective function, or violating the average and peak power constraints. Therefore, it always exists an optimal solution to (O-P2) in which $\tkm{}[n] \leq \log_2 ( 1 + \pm{} \snrr{}[n] )$ is  satisfied with equality. As a consequence, (O-P2) is  equivalent  to the original  problem, (O-P1).
\end{proof}

\section{Optimal R-UAVs to GUs Power Allocation}\label{App:P1}
For the sake of simplicity, to solve (O-P2.1), we first formulate the partial Lagrangian by taking into account the average rate and causality constraints. To this end, the partial Lagrangian of (O-P2.1)  is
\begin{multline}\label{eq:Lagrangian}
        \mathcal{L}( \mu, \pm{}, \tkm{}[n], \boldsymbol{\lambda}^{'}) = \mu \ls  +  \sum \limits_{k = 1}^K \lambda_k \bigg(  \sum \limits_{n = D+1}^N \sum \limits_{m = 1}^M \akm \tkm{}[n] - \mu \bigg) + \\ 
        \sum \limits_{n = D+1}^N \sum \limits_{m = 1}^M \lambda_{n,m} \bigg(  \sum \limits_{i = 1}^{n-D} \beta_{B,m}[i] R_{B,m}[i] -  \sum \limits_{i = D+1}^{n} \sum \limits_{k = 1}^{K} a_{m,k}[i] \tkm{}[i]   \bigg) ,
\end{multline}
\noindent where $\boldsymbol{\lambda}^{'} = \{ \lambda_k \ls \ls \forall k  \mspace{8mu} , \mspace{8mu} \lambda_{n,m} \ls \ls n = D+1,\dots,N \ls \ls \forall m \mspace{8mu} , \mspace{8mu}  \}$ represents the vector of partial Lagrangian multipliers. Then, we define the following variables:

\noindent\begin{minipage}{0.3\textwidth}
\begin{eqnarray}\label{eq:xLag}
    x = 1 - \sum \limits_{k = 1}^K \lambda_k
\end{eqnarray}
\end{minipage}%
\begin{minipage}{0.05\textwidth}\centering
    
\end{minipage}%
\begin{minipage}{0.65\textwidth}
\begin{eqnarray}\label{eq:wLag}
    \wkm{} = \lambda_k - \sum \limits_{i = n}^N \lambda_{i,m}  \ls \ls \ls \ls \ls n = D+1,\dots,N \ls \ls \forall m,k
\end{eqnarray}
\end{minipage}
\begin{eqnarray}\label{eq:zLag}
    \zmn{} = \sum \limits_{i = n+D}^N \lambda_{i,m}\ls \ls \ls \forall m, \ls \ls n = 1,\dots,N-D,
\end{eqnarray}
to re-write \eqref{eq:Lagrangian} as
\begin{multline*}
        \mathcal{L}( \mu, \pm{}, \tkm{}[n], \boldsymbol{\lambda}^{'}) = x \mu + \sum \limits_{k = 1}^K \sum \limits_{m = 1}^M \sum \limits_{n = D+1}^N  \wkm{} \akm \tkm{}[n] + \sum \limits_{m = 1}^M
        \sum \limits_{n = 1}^{N-D} \zmn{}  \beta_{B,m}[n] R_{B,m}[n] .
\end{multline*}

The goal is to maximize $\mathcal{L}( \mu, \pm{}, \tkm{}[n], \boldsymbol{\lambda}^{'})$ with the addition of the remaining constraints. To this end, the complete Lagrangian is given by
\begin{multline}\label{eq:Lagrangian3}
        \mathcal{L}( \mu, \pm{}, \tkm{}[n], \boldsymbol{\lambda}) = \mathcal{L}( \mu, \pm{}, \tkm{}[n], \boldsymbol{\lambda}^{'}) +  \sum \limits_{m = 1}^M \lambda_m^{'}( P_{R,avg} - \frac{1}{N}\sum \limits_{n = 1}^N \pm{}) + \\ \sum \limits_{n = D+1}^N \sum \limits_{m = 1}^M \sum \limits_{k = 1}^K \lambda_{n,m,k}^{'}\big( \log_2 ( 1 + \pm{} \Gamma_{m,k}^{''}[n] ) -  \tkm{}[n] \big) \sum \limits_{m = 1}^M \sum \limits_{n = 1}^N \lambda_{m,n}^{'}( P_{R,max} -  \pm{}).
\end{multline}
Taking the derivative with respect to the optimization variables for fixed values of the multipliers and recalling the binary nature of $\akm$,  we obtain
\begin{eqnarray}\label{eq:opttmk}
    \frac{ d \mathcal{L}( \mu, \pm{}, \tkm{}[n], \boldsymbol{\lambda}) }{d \tkm{}[n]} = 0 \xrightarrow[]{} \lambda_{m,n,k}^{'} = \wkm{} \akm
\end{eqnarray}
\begin{multline}\label{eq:optimalpower2}
    \frac{ d \mathcal{L}( \mu, \pm{}, \tkm{}[n], \boldsymbol{\lambda}) }{d \pm{}} = \sum_{k = 1}^K \frac{1}{\ln (2)}\frac{\lambda_{m,n,k}^{'} \Gamma_{m,k}^{''}[n]}{1 + \pm{\Gamma_{m,k}^{''}[n]}} - \frac{\lambda_m^{'}}{N} - \lambda_{m,n}^{'}  = 0  \\ \xrightarrow[]{} p_{m}^*[n] = \Big[ \frac{ \wkm{}  }{ (\lambda_{m,n}^{'} + \frac{\lambda_m^{'}}{N} ) \ln(2)} - \frac{1}{\Gamma_{m,k}^{''}[n]}  \Big]^+,
\end{multline}
\noindent where in the intermediate step of Eq. \eqref{eq:optimalpower2}, we have applied the solution obtained in \eqref{eq:opttmk} and therefore have only kept the term for which $\akm \neq 0$.

To obtain the Lagrangian multipliers, we aim to solve the dual problem. Actually, in order to have bounded solutions for  \eqref{eq:optimalpower2}, we need an extra constraint on $\boldsymbol{\lambda}$. Note that  in \eqref{eq:Lagrangian3}, if $\exists \ls \wkm{} < 0$,  its respective optimal $t_{m,k}^*[n] \xrightarrow{} - \infty $, making the problem unbounded. Therefore, $\wkm{} \geq 0$, which implies  $\lambda_k - \sum \limits_{i = n}^N \lambda_{n,m} \geq 0   \ls \ls n = D+1,\dots,N \ls \ls \forall m,k$. Hence, the Lagrangian multipliers are derived by solving the following convex optimization problem, named (O-P2.1D):
\begin{equation*}
    \begin{aligned}
    & \underset{\boldsymbol{\lambda}}{\text{min}}
    & & \mathcal{L}( \mu^*, p_{m}^{*}[n], t_{m,k}^{*}[n], \boldsymbol{\lambda}) \\
    & \text{s.t.} & &  \boldsymbol{\lambda} \geq \boldsymbol{0} \\
    & & & \lambda_k - \sum \limits_{i = n}^N \lambda_{i,m} \geq 0  \ls  \ls \ls \ls n = D+1,\dots,N \ls  , \ls \ls \forall m,k
    \end{aligned}
\end{equation*}

\noindent which can be efficiently solved by gradient methods. These two problems, (O-P2.1) and (O-P2.1D), are solved iteratively  until convergence. 

\section{Optimal MBS to R-UAVs Power Allocation}\label{App:P2}
To obtain the optimal power allocation from the MBS towards R-UAVs that minimizes the sum power at the MBS, we proceed in the same manner as in Appendix \ref{App:P1}. First, we form the complete Lagrangian and, afterwards, we find the optimal $\pb{}$ by making its derivative equal to zero. Therefore, the Lagrangian function of (O-P2.2) is given by
\begin{multline}\label{eq:Lagrangian2}
        \mathcal{L}( \mu, \pb{}, \boldsymbol{\lambda}) = \sum \limits_{n = 1}^{N-D}\pb{} -  \sum \limits_{m = 1}^M
        \sum \limits_{n = 1}^{N-D} \zmn{}  \beta_{B,m}[n] R_{B,m}[n] - \sum \limits_{m = 1}^M \lambda_m^{'}( P_{B,avg} - \frac{1}{N}\sum \limits_{n = 1}^N \pb{})- \\ \sum \limits_{m = 1}^M \sum \limits_{n = 1}^N \lambda_{m,n}^{'}( P_{B,max} -  \pb{}).
\end{multline}
Making the derivative equal to zero, we obtain
\begin{multline}\label{eq:optimalpower3}
    \frac{ d \mathcal{L}(  \pb{}, \boldsymbol{\lambda}) }{d \pb{}} = 1 - \sum_{m = 1}^M \frac{1}{\ln (2)}\frac{\zmn{} \bbm{} \Gamma_{B,m}^{''}[n]}{1 + \pm{}\Gamma_{B,m}^{''}[n]} + \frac{\lambda_m^{'}}{N} + \lambda_{m,n}^{'}  = 0  \\ \xrightarrow[]{} p_{B}^*[n] = \Big[ \frac{ \zmn{}  }{ (1 +\lambda_{m,n}^{'} + \frac{\lambda_m^{'}}{N} ) \ln(2)} - \frac{1}{\Gamma_{B,m}^{''}[n]}  \Big]^+,
\end{multline}
\noindent where in the intermediate step of Eq. \eqref{eq:optimalpower3}, we have used the fact that $\bbm{}$ is non-zero only for one R-UAV, 
and $\zmn{}$ is defined in the same manner as in \eqref{eq:zLag}. Finally, to obtain the Lagrangian multipliers, we use a gradient method to solve the dual problem, (O-P2.2D):
\begin{equation*}
    \begin{aligned}
    & \underset{\boldsymbol{\lambda}}{\text{max}}
    & & \mathcal{L}( p_{B}^*[n] , \boldsymbol{\lambda}) \\
    & \text{s.t.} & &  \boldsymbol{\lambda} \geq \boldsymbol{0}
    \end{aligned}
\end{equation*}
Again, we iterate between (O-P2.2) and (O-P2.2D) until convergence.

\section{Proof of Proposition 1}\label{App:ProofProp1}
\begin{proof}
We distinguish between the outer Iteration $k$ in Algorithm 3, and the inner Iterations $j$ and $i$ of Algorithms 1 and 2, respectively. At Iteration $k$, for given powers $\P^k$, Algorithm 1 is composed of three convex problems whose global solutions can be attained. Therefore, iterating on the inner variable $j$ provides the following inequalities: (i) $\mu(\A^{k,j},\Q^{k,j},\R^{k,j},\P^k) \leq \mu(\A^{k,j+1},\Q^{k,j},\R^{k,j},\P^k)$ by solving the LP problem (O-X1),  (ii) applying the SCP technique to the UAV trajectory sub-problem (O-T2) provides the following:  $\mu(\A^{k,j+1},\Q^{k,j},\R^{k,j},\P^k) \leq \mu(\A^{k,j+1},\Q^{k,j+1},\R^{k,j},\P^k)$ and  (iii)  optimizing the directivity degrees, i.e.,  the convex optimization problem (O-R2), results in: $\mu(\A^{k,j+1},\Q^{k,j+1},\R^{k,j},\P^k) \leq \mu(\A^{k,j+1},\Q^{k,j+1},\R^{k,j+1},\P^k)$. As a result Algorithm 1 provides a non-decreasing sequence: $\mu^{k,1} \leq \mu^{k,2} \leq   \dots \leq \mu^{k,j} \dots \mu^{k,*} $. Applying the same procedure for Algorithm 2, given the inner iteration $i$ and fixing $\A^k$, $\Q^k$ and $\R^k$, the following inequality holds: $\mu(\A^{k},\Q^{k},\R^{k},\P^{k,i}) \leq \mu(\A^{k},\Q^{k},\R^{k},\P^{k,i+1})$ since (O-P2.1) and (O-P2.2) are convex problems. As a result, Algorithms 1 and 2 provide non-decreasing values on the objective function when iterating in their respective inner variables. Therefore, the sequence of achieved GU rates, is a non-decreasing sequence in the outer variable as well: $\mu^{1} \leq \mu^{2} \leq \dots \leq \mu^k \leq \dots \leq \mu^{*}  $, where for simplicity $ \mu^k = \mu(\A^{k},\P^k,\Q^{k},\R^k)$ is the objective function given the solution at Iteration $k$ after convergence in $j$ and $i$. Since the minimum GU rates are upper-bounded by $ \mu_{max} =  \log_2 \big( 1 + \frac{P_{R,max} \rhoo{} (r_{max} + 1) }{ \sigma^2 \hr^\kappa } \big)$, Algorithm 3 will converge, as well as Algorithms 1 and 2.
\end{proof}

\ifCLASSOPTIONcaptionsoff
  \newpage
\fi

\bibliography{Main}
\bibliographystyle{ieeetr}



\end{document}